\documentclass[]{article}
\usepackage{filecontents}
\usepackage [english] {babel}
\usepackage{amsmath,amsfonts,amssymb,amsthm,epsfig,epstopdf,titling,array}
\usepackage{ulem}
\usepackage{color}
\usepackage [utf8]{inputenc}
\usepackage[T1]{fontenc}
\usepackage{a4wide}

\usepackage[margin=10pt, font=scriptsize, labelfont=bf]{caption}

\usepackage{enumerate}

\usepackage[colorlinks]{hyperref}

\renewcommand{\d}{\mathrm d}

\hypersetup{colorlinks=true,linkcolor=blue,citecolor=blue}
\numberwithin{equation}{section}

\usepackage{tikz}
\usepackage{subfigure}

\tikzset{vertex/.style={circle,fill=black,inner sep=2pt},
	ctVertex/.style={diamond,fill=black,inner sep=2.5pt},
	bigvertex/.style={circle,fill=black,inner sep=4pt},
	E/.append style={fill=white,draw},
	nuEP/.style={circle,fill=white,draw, inner sep=2pt},
	linelabel/.style={sloped,above,very near start, inner sep=1pt,execute at begin node=$\scriptstyle,execute at end node=$},
	baseline=(current  bounding  box.center),doubled/.style={double distance= 1pt,line width=1.5pt}
}

\newtheorem{theorem}{Theorem}[section]
\newtheorem{lemma}[theorem]{Lemma}
\newtheorem{prop}[theorem] {Proposition}

\newtheorem{definition}[theorem] {Definition}

\theoremstyle{definition}

\theoremstyle{remark}
\newtheorem{remark}{Remark}

\numberwithin{remark}{section}

\newcommand{\card}[1]{\left|#1\right|}

\usepackage{titling}
\thanksmarkseries{arabic}

\title{Free energy expansions for renormalized systems for colloids}
\author{Tong Xuan Nguyen\thanks{New York University Shanghai, 567 W Yangsi Rd, Pudong New Area, Shanghai, China 200126, \texttt{tn2137@nyu.edu}}, 
	Giuseppe Scola\thanks{Scuola Internazionale Superiore di Studi Avanzati (SISSA), via Bonomea 265, 34136 Trieste, Italy, \texttt{gscola@sissa.it}}, 
	Dimitrios Tsagkarogiannis\thanks{Dipartimento di Ingegneria e Scienze dell'Informazione e Matematica, Universit\`a degli Studi dell'Aquila, 67100 L'Aquila, Italy, \texttt{dimitrios.tsagkarogiannis@univaq.it}}}

\begin{document}

	\date{}
	\maketitle
	\begin{abstract} We consider a binary system of small and large spheres of finite size 
		in a continuous medium interacting via a non-negative potential.
		We work in the canonical ensemble and compute upper and lower bound  for the free energy at finite and infinite volume by first integrating over the small spheres and then treating the effective system of the large ones which now interact via a multi-body potential.
		By exploiting the underlying structure of the original binary system we 
		prove the convergence of the cluster expansion for the latter system
		and obtain a sufficient condition which involves the surface of the large spheres
		rather than their volume (as it would have been the case in a direct application of existing methods directly to the binary system). 
		Our result is valid for the particular case of hard spheres (colloids) for which we rigorously
		treat the depletion interaction.\\
		
		
	\end{abstract}

	\tableofcontents

	\section{Introduction}

	We consider a system of a given number of small and big spheres 
	(of finite radii $r$ and $R$ respectively) all interacting via hard-core.
	Our initial motivation comes from the study of the
	droplet picture of condensation: increasing the density of the system, as we approach the
	regime of phase coexistence we observe the formation
	of one or more big droplets (of macroscopic size, representing the ``solid'')
	co-existing with some microscopic-size clusters representing the ``liquid''.
	Their shape can vary, but with a further crude simplification we consider them to
	be spherical of any size ranging from microscopic (for the liquid) to macroscopic (for the solid). By a further simplification we consider only two sizes: $N_r$ many microscopic spheres of radius $r$ and $N_R$ many macroscopic ones
	of radius $R$ such that
	\begin{equation}\label{size}
		N_R |B_R|\sim |\Lambda|,
	\end{equation}
	where $|B_R|$ is the volume of a (big) sphere of radius $R$.
	Hence, in this paper we first consider the
	case of finite radii $r$ and $R$ at finite and infinite volume, but our next goal is to study
	the infinite volume with an additional scale separation of $R/r\to\infty$.
	A second motivation comes from the colloidal systems where large objects
	are subject to effective interactions mediated by small particles.
	
	In the grand-canonical ensemble, 
	by a straightforward application of the standard cluster expansion method one
	can show that the pressure (as a function of the activities $z_r$ and $z_R$) 
	can be written as an absolutely convergent series
	with a radius of convergence that depends on the excluded volume between the spheres (small and big).
	By a more careful analysis, in \cite{jansen2020cluster} it has been proved that 
	if we first integrate over the small
	spheres and then study the effective (multi-body) system we obtain an improved radius of convergence. This is based on the fact that
	the interaction of two big spheres via a small one is happening
	due to the overlap
	of a big and a small without the first big sphere touching the second big. 
	Hence the overlap of the small with the big is limited at a corridor of breadth $r$ of the boundary of the big sphere and thus it is of the order of the surface of the big sphere rather than its volume. Consequently, 
	it is better manifested in a further limit $R/r\to\infty$ and
	in order to obtain it we need to consider 
	the effective activity 
	of a big sphere and all possible interactions with the small ones.
	This is also relevant in colloidal systems where as a result of the concentration
	of the small spheres the effective interaction between the big ones can be tuned properly. A typical phenomenon is  depletion attraction \cite{asakura1958interaction}, \cite{tuinier2011colloids}.
	Hence, in this paper we show that such a picture is true if we work directly in the canonical ensemble with tuning parameters the density of the spheres. We show that for fixed radii $r$ and $R$ we can first integrate (as in the
	grand-canonical ensemble) over the small spheres and then treat the multi-boby effective interaction
	between the big ones. However, the difference in the canonical ensemble case is 
	that we can only give
	upper and lower bounds for the free energy as in Theorem~\ref{ClusterEffective}.
	Their discrepancy seems negligible in the further limit $R/r\to\infty$, see Remark~\ref{r:lim}, the regime for which the improvement is more transparent.
	But, this requires further analysis which goes beyond the goal of the present paper
	and it is also related to other interesting phenomena such as phase transitions in binary colloidal hard-sphere mixtures, see \cite{dijkstra1999phase}, \cite{kobayashi2021critical} and \cite{lopez2013communication}. 
	
	Another interesting issue which is worth of further investigation is the following: 
	if we treat the problem as a two-species system, the virial inversion (e.g. by applying \cite{jansen2014multispecies}) of the pressure in $z_r$ and $z_R$ would give a power series expansion in
	the densities $\rho_r$ and $\rho_R$:
	\begin{equation}\label{can2spec}
		\sum_{n,m\geq 0}\rho_R^n \rho_r^m \beta_{n,m},
	\end{equation}
	where $\beta_{n,m}$ are the coefficients expressed via two-connected graphs
	on two-types of vertices (corresponding to the small and big spheres). 
	The radius of convergence will depend again on the excluded
	volume between the spheres. Equivalently, one can also perform
	a direct proof in the canonical ensemble for two-species which would give a similar
	radius of convergence by a simple adaptation of the proof of \cite{pulvirenti2012cluster}.
	However, the intriguing observation is that in the canonical ensemble we do not need to change variables as
	in the grand-canonical case. This is due to the fact that 
	\begin{equation*}
		\rho_R:=z_R\frac{\partial}{\partial z_R}p(z_r,z_R)=\hat z_R\frac{\partial}{\partial \hat z_R}\hat p(z_r,\hat z_R(z_r)), \qquad \hat p(z_r,\hat z_R(z_r))=p(z_r,z_R),
	\end{equation*}
	where $\hat z_R(z_r)$ is the new activity of the big spheres if we first integrate over
	the small ones. 
	We observe that while we have two options for the activities $z_R$ and $\hat z_R$, both correspond to the same density $\rho_R$.
	In this paper we prove that if, working in the canonical ensemble, we perform a similar
	two-step integration (first the small ones and then the big with multi-body effective interactions) we obtain the improved convergence condition. This is essentially due to the fact that we use the same (improved) tree-graph inequality as in the grand-canonical \cite{jansen2020cluster}. Hence, one would be temped to view it as an {\it example where the direct canonical expansion gives a power series in $\rho_r,\rho_R$ with a better radius of convergence than the virial inversion of the power series in the dual variables $z_r, z_R$}.
	The catch is that what we really obtain is a reordering of the power series \eqref{can2spec}
	of the two-species model. In fact, instead of the form \eqref{can2spec} we obtain
	\begin{equation}\label{can}
		\sum_{n\geq 0}\rho_R^n B^*(n;\rho_r),
	\end{equation}
	where $B^*(n;\rho_r)$ is again an infinite sum in powers of $\rho_r$. Hence, in this formulation we have the absolute value outside infinite sums, instead
	of its previous position over all two connected graphs.
	Notice that if we expanded the ``inside'' infinite sum we would obtain the necessary cancellations towards the two connected graphs, to the expense of obtaining a worse radius of convergence (since expanding corresponds to bringing the absolute value more inside the sum and hence summing more positive terms).
	Alternatively, one should be able to obtain this rearranged free energy expansion by a virial inversion
	of the improved grand-canonical expansion, but it is still not clear to us how to do it. 
	
	Summarizing, in Section~\ref{s_results} we present the various contributions for the upper and lower bound of the finite volume canonical free energy (Theorem~\ref{ClusterEffective}). 
	In Section~\ref{s_small} we show how to perform the canonical ensemble cluster expansion for the small spheres in the presence of the big ones.
	The emerging effective multi-body interaction between the big spheres
	is treated in Section~\ref{s_big}.
	Last, in Section~\ref{s:canc} we study various cancellations that allow to pass to the thermodynamic limit and at the same time reveal interesting graphical representations as summarized in Figure~\ref{fig:4}.
	Some intriguing questions arise: what happens in the limit
	$R/r\to\infty$ where it is expected that the two bounds agree? Can we obtain it
	via virial inversion from the grand canonical ensemble as well or is it only accessible from the canonical?

	\section{The model and main results}\label{s_results}

	We consider a system composed of $N_r$ small and $N_R$ big particles
	in a box $\Lambda:=\left(-\frac{L}{2},\frac{L}{2}\right]^d\subset\mathbb{R}^3$ of length $L>0$ interacting via a hard-core potential with periodic boundary conditions. 
	Denoting by $p,p'\in\Lambda$ the positions of two big particles and by
	$q,q'\in\Lambda$ the position of two small ones, we consider the following interaction potentials:
	\begin{equation}
		V_{ll}(p-p'):=\begin{cases}
			+\infty&{\rm{if}}\;|p-p'|<2R\\
			0&{\rm{otherwise}}
		\end{cases},
	\end{equation}	
	\begin{equation}
		V_{ss}(q-q'):=\begin{cases}
			+\infty&{\rm{if}}\;|q-q'|<2r\\
			0&{\rm{otherwise}}
		\end{cases},
	\end{equation}	
	and
	\begin{equation}
		V_{ls}(p-q):=\begin{cases}
			+\infty&{\rm{if}}\;|p-q|<R+r\\
			0&{\rm{otherwise}}
		\end{cases},
	\end{equation}	
	where $0<R,r<L$ are the radii of the big and small spheres.
	We define the Hamiltonian
	\begin{equation}
		\label{Hamiltonian}
		H_{\Lambda}(\mathbf{x}):=\sum_{1\leq i<j\leq N_R} V_{ll}(p_i-p_j)
		+ \sum_{1\le i<j\le N_r} V_{ss}(q_i-q_j)+\sum_{\substack{1\le i\le N_R \\ 1\leq j\leq N_r} }V_{ls}(p_i-q_j),
	\end{equation}
	where $\mathbf{x}=(p_1,...,p_{N_R},q_1,\ldots,q_{N_r})$ are the centers of the $N_R$ big and $N_r$ small spheres.

	The canonical partition function at inverse temperature $\beta>0$ is given by
	\begin{equation}\begin{split}\label{CPF}
			Z_{\Lambda,\beta,N_R,N_r}&:=\frac{1}{N_r!}\frac{1}{N_R!}\int_{\Lambda^{N_R}}	\prod_{1\le i<j\le N_R}[1+f^{ll}(p_i,p_j)]
			\times\\
			&\hspace{.5cm}\int_{\Lambda^{N_r}}\prod_{\substack{1\le j\le N_R\\1\le i\le N_r}}[1+f^{ls}(p_j,q_i)]
			\prod_{1\le i<j\le N_R}[1+f^{ss}(q_i,q_j)\prod_{i=1}^{N_r}\frac{\d q_i}{|\Lambda|}
			\prod_{i=1}^{N_R}\frac{\d p_i}{|\Lambda|},
		\end{split}
	\end{equation}	
	where we have defined:
	\begin{equation}\begin{split}
			\label{MayerFunction}	
			f^{ll}(p_i,p_j) & := e^{-\beta V_{ll}(p_i-p_j)}-1=-\mathbf{1}_{\{|p_i-p_j|<2R\}},\\
			f^{ss}(q_i,q_j) &:=e^{-\beta V_{ss}(p_i-p_j)}-1=-\mathbf{1}_{\{|q_i-q_j|<2r\}},\\
			f^{ls}(p_i,q_j) &:=e^{-\beta V_{ls}(p_i-q_j)}-1=-\mathbf{1}_{\{|p_i-q_j|<R+r\}}.
	\end{split}\end{equation} 
	Following \cite{pulvirenti2012cluster} we will view our system as a perturbation around the non-interacting case. From now on we remove the dependence on $\beta$ for a simpler notation. We write
	\begin{equation*}
		Z_{\Lambda, N_R,N_r}=Z^{ideal}_{\Lambda,N_R,N_r}Z^{int}_{\Lambda,N_R,N_r},
	\end{equation*}	 
	where
	\begin{equation}\label{Zideal}
		Z^{ideal}_{\Lambda,N_R,N_r}:=\frac{|\Lambda|^{N_R}}{N_R!}\frac{|\Lambda|^{N_r}}{N_r!},
	\end{equation}
	and
	\begin{equation}\begin{split}\label{Zint}
			Z^{int}_{\Lambda,N_R,N_r}  := & \int_{\Lambda^{N_R}}	\prod_{1\le i<j\le N_R}[1+f^{ll}(p_i,p_j)]
			\times\\
			&
			\int_{\Lambda^{N_r}}\prod_{\substack{1\le j\le N_R\\1\le i\le N_r}}[1+f^{ls}(p_j,q_i)]
			\prod_{1\le i<j\le N_R}[1+f^{ss}(q_i,q_j)\prod_{i=1}^{N_r}\frac{\d q_i}{|\Lambda|}
			\prod_{i=1}^{N_R}\frac{\d p_i}{|\Lambda|}.
	\end{split}\end{equation}	
	The goal of the paper is to write the finite volume free energy
	\begin{equation}\label{FVfreeEn}
		f_{\Lambda,\beta}(N_R,N_r):=-\frac{1}{\beta |\Lambda|}\log Z^{per}_{\Lambda,N_R,N_r},
	\end{equation}
	as an absolutely convergent power series (uniformly on its parameters)
	and discuss possible infinite volume limits, such as
	\begin{equation}\label{TDfreeEn}
		f_\beta(\rho_R,\rho_r):=\lim_{\substack{\Lambda\to\mathbb{R}^d,\,
				N_R, N_r\to\infty,
				\\
				N_R=\lfloor\rho_R|\Lambda|\rfloor,\,
				N_r=\lfloor\rho_r|\Lambda|\rfloor
		}}		f_{\Lambda,\beta}(N_R,N_r),
	\end{equation}	
	for some values of $\rho_r,\rho_R>0$.
	
	As we described in the introduction we have two main options. The first is to expand the products in \eqref{CPF} and obtain a two-colour graph. Then one can directly apply the method developed in \cite{pulvirenti2012cluster} obtaining the free energy as a power series expansion in the
	densities $N_r/|\Lambda|$ and $N_R/|\Lambda|$ in some disc which depends on the excluded volumes
	$|B_{2R}|$, $|B_{R+r}|$ and $|B_{2r}|$.
	Note that with $B_R(0)$ we denote a sphere of radius $R$ and center $0$. By $|B_R|$ we denote its volume.
	Alternatively, one can proceed as in \cite{jansen2020cluster} and perform a two-step integration
	by first integrating over the small particles obtaining an effective model with multi-body
	interactions between the large ones.
	This is the context of Theorem~\ref{ClusterEffective}. In Remark~\ref{key_remark} 
	we show that in this fashion we obtain an improved
	convergence condition, as it was the case in \cite{jansen2020cluster} for the grand-canonical ensemble. 
	In principle, by virial inversion one should obtain the finite volume grand-canonical free energy (which should agree with the canonical one at the infinite volume limit), but performing it for the two-step integration
	it turned out to be more complicated than expected and we leave it for a forthcoming work.
	
	In order to proceed with this plan we define:
	\begin{eqnarray}\label{Zps}
		Z^{\mathbf{p}}_{\Lambda,N_r}&:=&\int_{\Lambda^{N_r}}\prod_{1\le i< j\le N_r}\left[1+f^{ss}(q_i,q_j)\right]\prod_{\substack{1\le i\le N_R\\ 1\le j\le N_r}}\left[1+f^{ls}(p_i,q_j)\right]\prod_{i=1}^{N_r}\frac{\d q_{i}}{|\Lambda|},\nonumber\\
	\end{eqnarray}
	representing the canonical partition function for the small spheres
	in the presence of the big ones with centers $\mathbf{p}:=(p_1,...,p_{N_R})$.
	The resulting formula can be viewed as an effective multi-body interaction potential 
	between the big spheres.
	For a cluster expansion formulation we could follow
	the strategy in \cite{pulvirenti2012cluster}, but for that
	we need to normalize the prior measure by the available volume of the small spheres
	for a fixed configuration of the big ones:
	\begin{equation}\label{lambdatilde}
		\tilde\Lambda(\mathbf p):=\Lambda\setminus \cup_{j=1}^{N_R}B_{R+r}(p_j),
	\end{equation}
	with
	\begin{equation}\label{abslambdatilde}
		|\tilde\Lambda(\mathbf p)|=\int_{\Lambda}\prod_{j=1}^{N_R}(1+f^{ls}(p_j,q))\d q\geq 0.
	\end{equation}
	We remark that $\tilde\Lambda(\emptyset)=\Lambda$.
	A key observation is the dependence of the normalization $|\tilde\Lambda(\mathbf p)|$ on the whole configuration $\mathbf p$ which introduces an additional (to $Z^{\mathbf{p}}_{\Lambda,N_r}$) $\mathbf p$-dependent 
	factor in a subsequent integration over the big spheres.
	Nevertheless, we have that
	\begin{equation}\label{estimate}
		|\Lambda|-N_R |B_{R+r}|
		\leq
		|\tilde\Lambda(\mathbf p)|
		\leq
		|\Lambda|-N_R |B_{R}|,
	\end{equation}
	hence, in the sequel the following ratio will be relevant
	\begin{equation}\label{ratio}
		\frac{1}{1-\frac{N_R}{|\Lambda|}|B_R|}\leq\frac{|\Lambda|}{|\tilde\Lambda(\mathbf p)|}=\frac{1}{1+\int_{\Lambda}\left[\prod_{j=1}^{N_R}(1+f^{ls}(p_j,q))-1\right]\frac{\d q}{|\Lambda|}}
		\leq \frac{1}{1-\frac{N_R}{|\Lambda|}|B_{R+r}|}.
	\end{equation}
	We notice two facts: first, the difference between the upper and lower bound of \eqref{ratio}
	vanishes in the limit $R/r\to\infty$. Second, the above bounds are not dependent on the actual configuration $\mathbf p$, but rather only on the total number $N_R$.
	
	We decompose the free energy \eqref{FVfreeEn} as follows:
	\begin{equation}\label{EffectiveFVfreeE}
		\beta f_{\Lambda,\beta}(N_R,N_r)=-\frac{1}{|\Lambda|}\left(\log Z^{ideal}_{\Lambda,N_R,N_r}+\log Z^{\emptyset}_{\Lambda,N_r}+\log \hat{Z}_{\Lambda,N_R, N_r}\right).	
	\end{equation}	
	The entropic part $Z^{ideal}_{\Lambda,N_R,N_r}$ is given in \eqref{Zideal}.
	The second term is the contribution when no big spheres are involved
	\begin{equation}\label{CPFonlyS}
		Z^{\emptyset}_{\Lambda,N_r}:=\int_{\Lambda^{N_r}}\prod_{1\le k<l\le N_r}\left[1+f^{ss}_{k,l}\right]\prod_{i=1}^{N_r}\frac{\d q_i}{|\Lambda|},
	\end{equation}	
	where we assume periodic boundary conditions.
	This has been studied in \cite{pulvirenti2012cluster}, Theorem 1, which states that there exists a constant $\rho^*>0$ such that if 
	\begin{equation}\label{Condition1}
		\frac{N_r}{|\Lambda|}|B_{2R}|\le \rho^*.
	\end{equation}
	Then
	\begin{eqnarray}\label{ClusterOnlySmall}
		\frac{1}{|\Lambda|}\log Z^{\emptyset}_{\Lambda,N_r}=\frac{N_r}{|\Lambda|}\sum_{n\ge 1}\frac{1}{n+1}P_{|\Lambda|,N_r}(n)B^{\emptyset}_{\Lambda}(n),
	\end{eqnarray}
	where
	\begin{equation}\label{P}
		P_{|\Lambda|,N_r}(n):=\begin{cases}
			\displaystyle\frac{(N_r-1)\cdot\cdot\cdot(N_r-n)}{|\Lambda|^{n}}&{\rm{if}}\; n<N_r \\
			\displaystyle 0&{\rm{otherwise}}	
		\end{cases},
	\end{equation}  
	and $B^{\emptyset}_\Lambda(n)$ will be given in \eqref{B_empty}. 
	In the thermodynamic limit we obtain
	\begin{equation}\label{lim}
		\lim_{\substack{N_r,|\Lambda|\to\infty \\ N_r=\lfloor\rho_r|\Lambda|\rfloor}}
		P_{|\Lambda|,N_r}(n)B^{\emptyset}_{\Lambda}(n)
		=\rho_r^{n}\beta_n,
	\end{equation}  
	for all $n\geq 1$. Note that $\beta_n$ are the ``irreducible coefficients'', given by
	\begin{equation}\label{betan}
		\beta_n=\frac{1}{n!}\sum_{g\in\mathcal B_{n+1}}\int_{(\mathbb R^d)^n}\prod_{(i,j)\in E(g)}
		f^{ss}(q_i,q_j)\d\mathbf q_{[n]},\qquad q_{n+1}\equiv 0,
	\end{equation}
	where $[n]:=\{1,\ldots,n\}$ and $\mathcal B_n$ is the set of two-connected graphs on $n$ vertices. We recall some basic definitions:
	\begin{definition}
		A cut-point of a connected graph is a vertex whose removal yields a disconnected graph.
		An irreducible graph is free of cut-points.
	\end{definition}
	
	Thus, the main task in this paper is
	to compute the finite volume free energy of the effective system of big spheres after the integration over the small
	(third term in \eqref{EffectiveFVfreeE}): 
	\begin{equation}\label{EffectiveCPF}
		\hat{Z}_{\Lambda,N_R, N_r}:=\frac{Z^{int}_{\Lambda,N_R,N_r}}{Z^{\emptyset}_{\Lambda,N_r}}
		=
		\frac{1}{Z^{\emptyset}_{\Lambda,N_r}}
		\int_{\Lambda^{N_R}}\prod_{1\le i<j\le N_R}[1+f^{ll}(p_i,p_j)]Z^{\mathbf{p}}_{\Lambda,N_r}\prod_{i=1}^{N_R} \frac{\d p_i}{|\Lambda|}.
	\end{equation}	 
	Using \eqref{ratio} we provide upper and lower bounds, presented in Proposition~\ref{ClusterEffective2}.
	Hence, overall, following the decomposition in \eqref{EffectiveFVfreeE}
	we obtain upper and lower bounds
	for $\displaystyle\frac{1}{|\Lambda|}\ln\hat{Z}_{\Lambda,N_R, N_r}$
	given in the following theorem:
	\begin{theorem}\label{ClusterEffective}
		Suppose that for some $a,b,c\ge0$ the following conditions hold:
		\begin{equation}
			e^{b+c}|B_{R+r}(0)\setminus B_{R-r}(0)|\frac{N_r}{|\Lambda|-N_R |B_{R+r}|} +e^a|B_{2R}|\frac{N_R}{|\Lambda|}\le a
		\end{equation}
		\begin{equation}
			e^a|B_{R+r}(0)\setminus B_{R-r}(0)|\frac{N_R}{|\Lambda|}\le b,
		\end{equation}
		as well as \eqref{cond1}.
		
		Then, we have upper and lower bounds for $\displaystyle\frac{1}{|\Lambda|}\log \hat{Z}_{\Lambda,N_R,N_r}$ given by
		\begin{equation}\begin{split}\label{new_ub}
				\frac{1}{|\Lambda|}\log \hat{Z}^u_{\Lambda,N_R,N_r}
				& = 
				\frac{N_r}{|\Lambda|}\sum_{k\geq 1}\frac{1}{k+1}
				P_{N_r,|\Lambda|}(k)A_{\Lambda}(k;\frac{N_R |B_{R+r}|}{|\Lambda|})
				+\\
				&
				+\frac{N_R}{|\Lambda|}\sum_{k\geq 1}\frac{1}{k+1}
				\left(\frac{N_r}{|\Lambda|-N_R |B_{R+r}|}\right)^kB^{*,(1)}_{\Lambda}(k)\\
				& +
				\frac{N_R}{|\Lambda|}\sum_{n\ge1}P_{|\Lambda|,N_R}(n)B^*_{\Lambda}(n, N_r)
				+O\left(\frac{1}{|\Lambda|}\right)
		\end{split}\end{equation}
		and
		\begin{equation}\begin{split}\label{new_lb}
				\frac{1}{|\Lambda|}\log \hat{Z}^l_{\Lambda,N_R,N_r}
				& = 
				\frac{N_r}{|\Lambda|}\sum_{k\geq 1}\frac{1}{k+1}
				P_{N_r,|\Lambda|}(k)A_{\Lambda}(k;\frac{N_R |B_R|}{|\Lambda|})
				+\\
				&
				+\
				\frac{N_R}{|\Lambda|}\sum_{k\geq 1}\frac{1}{k+1}
				\left(\frac{N_r}{|\Lambda|-N_R |B_{R}|}\right)^k B^{*,(1)}_{\Lambda}(k)\\
				&+
				\frac{N_R}{|\Lambda|}\sum_{n\ge1}P_{|\Lambda|,N_R}(n)B^*_{\Lambda}(n;N_r)
				+O\left(\frac{1}{|\Lambda|}\right),
		\end{split}\end{equation}
		where $A_{\Lambda}$ is given in \eqref{ALambda}, $B^{*,(1)}_\Lambda$
		in \eqref{B1} and 
		$B^*_{\Lambda}$ in \eqref{withHG}.
	\end{theorem}
	
	\begin{proof}
		The proof is a direct consequence of Proposition~\ref{ClusterEffective2} and Lemmas~\ref{second_cancellation1} and \ref{LemmaUltimo}.
	\end{proof}
	
	\begin{remark}\label{zero_remark}
		At the end of Section~\ref{s:canc} we give the thermodynamic limit of the above upper and lower bounds.
		The two bounds differ due to the slightly different entries in the first two terms.
		In Remark~\ref{r:lim} we comment that this discrepancy may be negligible in the further limit $R/r\to\infty$,
		which however requires more efforts for a complete analysis.
		Furthermore, some of these terms may cancel but only
		under a weaker convergence condition, see Remark~\ref{r:canc}. This is a key point of this paper and we
		will discuss it in detail.
	\end{remark}

	\section{Cluster expansion for the small spheres}\label{s_small}

	Next we give a precise formulation of the quantities considered in the previous
	expository section. We start with the computation of $Z^{\mathbf{p}}_{\Lambda,N_r}$.
	In \cite{jansen2020cluster} the cluster expansion has been performed in the grand-canonical ensemble. However,	following the main idea in \cite{pulvirenti2012cluster}, 
	a similar expansion is possible in the canonical ensemble as well.
	Let
	\begin{equation}\label{PolymerSmall}
		\mathcal{V}_r:=\{V\subset[N_r]\;:\;|V|\ge 2\}
	\end{equation} 
	be the space of subsets (of cardinality at least two) of the labels of the small spheres.
	We write:
	\begin{equation}\begin{split}\label{att}
			Z^{\mathbf{p}}_{\Lambda,N_r}& =\frac{1}{|\Lambda|^{N_r}}\times\\
			& 
			\int_{\Lambda^{N_r}}\sum_{\substack{\{V_1,\ldots,V_n\}\\ \text{part. of } [N_r]}}\prod_{i=1}^n 
			\left\{\sum_{g\in\mathcal C_{|V_i|}}\prod_{\{k,l\}\in E(g)}f^{ss}(q_k,q_l)
			\prod_{k\in V_i}
			\prod_{j=1}^{N_R}(1+f^{ls}(p_j,q_k)) \d\mathbf q_{V_i}\right\}\\
			& =  \left(\frac{|\tilde\Lambda(\mathbf p)|}{|\Lambda|}\right)^{N_r}\left(
			1+\sum_{n\geq 1}\frac{1}{n!}
			\sum_{\substack{(V_1,\ldots,V_n)\\ V_i\in\mathcal V_r}}
			\varphi(V_1,\ldots,V_n)
			\prod_{i=1}^n\zeta^{\mathbf p}_{\Lambda}(V_i)
			\right),
	\end{split}\end{equation}
	where 
	\begin{equation}\label{phi}
		\varphi(V_1,\ldots,V_n):=\prod_{1\le i<j\le n}\mathbf 1_{V_i\cap V_j=\emptyset}.
	\end{equation}
	Here $\mathcal{C}_n$ denotes the connected graph set on $n$ vertices.
	The activity $\zeta^{\mathbf p}_{\Lambda}$ is defined by: 
	\begin{equation}\label{zetal}
		\zeta^{\mathbf p}_{\Lambda}(V):=
		\left(\frac{|\Lambda|}{|\tilde\Lambda(\mathbf p)|}\right)^{|V|}
		\int_{\Lambda^{|V|}}\sum_{g\in\mathcal C_{|V|}}\prod_{i,j}f^{ss}(q_i,q_j)
		\left(1+\vartheta_{\mathbf{p}}(\mathbf q_V)\right)\prod_{i\in V}\frac{\d q_{i}}{|\Lambda|},
	\end{equation}
	where $\mathbf q_{V}\equiv (q_i)_{i\in V}$ and
	\begin{equation}\label{T}
		1+\vartheta_{\mathbf{p}}(\mathbf q_V):=\prod_{k\in V}\prod_{j=1}^{N_R}\left[1+f^{ls}(p_j,q_k)\right]\geq 0.
	\end{equation}	
	Recall that $|\tilde\Lambda(\mathbf p)|$
	is given in \eqref{abslambdatilde}.
	Note that in order to go from the first to the second line in \eqref{att} we had to normalize the measure with the quantity $|\tilde\Lambda(\mathbf p)|$
	so that $\zeta^{\mathbf p}_{\Lambda}(V)=1$, if $|V|=1$. In this manner, 
	we remove the constraint that
	the collection of $V_1,\ldots,V_n$ is a partition (as in the second line).
	Thanks to that, we are in the context of the Abstract Polymer Model \cite{kotecky1986cluster} and we can perform the cluster expansion.
	As mentioned before, 
	a key observation is the dependence of the normalization $|\tilde\Lambda(\mathbf p)|$ on the whole configuration $\mathbf p$.
	However, the bounds \eqref{estimate}
	depend only
	on the total number $N_R$.
	Hence, in the next section we will develop the cluster expansion over the configuration
	$\mathbf p$ taking into account these two bounds, 
	rather than the more complicated $\mathbf p$-dependent
	factor.
	We have:
	\begin{equation}\begin{split}\label{calc}
			\frac{1}{|\Lambda|}\ln Z^{\mathbf{p}}_{\Lambda,N_r} & = 
			\frac{1}{|\Lambda|}\ln\left(\frac{|\tilde\Lambda(\mathbf p)|}{|\Lambda|}\right)^{N_r}\\
			&+\frac{1}{|\Lambda|}
			\ln\left(
			1+\sum_{n\geq 1}\frac{1}{n!}
			\sum_{\substack{(V_1,\ldots,V_n)\\ V_i\in\mathcal V_r},\, |V_i|\geq 2}\varphi(V_1,\ldots,V_n)
			\prod_{i=1}^n\zeta^{\mathbf p}_{\Lambda}(V_i)
			\right)\\
			& = 
			\frac{N_r}{|\Lambda|}\ln\frac{|\tilde\Lambda(\mathbf{p})|}{|\Lambda|}
			+\frac{1}{|\Lambda|}
			\sum_{n\geq 1}\frac{1}{n!}
			\sum_{\substack{(V_1,\ldots,V_n)\\ V_i\in\mathcal V_r,\, |V_i|\geq 2}}
			\varphi^T(V_1,\ldots,V_n)
			\prod_{i=1}^n\zeta^{\mathbf p}_{\Lambda}(V_i),
	\end{split}\end{equation}
	where
	\begin{equation}\label{phiT}
		\varphi^T(V_1,\ldots,V_n):=\sum_{g\in\mathcal {C}_n}\prod_{\{i,j\}\in E(g)}-\mathbf 1_{V_i\cap V_j\neq\emptyset}.
	\end{equation}
	The expression is an absolutely convergent power series with the absolute value at the activities $\zeta^{\mathbf p}_{\Lambda}$. The proof is straightforward and will be given in the appendix, see Lemma~\ref{PropXuan2}.
	
	\begin{remark}\label{r:APM}
		We observe that
		\begin{equation}\label{comp_act}
			\prod_{i\in V}\prod_{j=1}^{N_R}(1+f^{ls}(p_j,q_i))=
			\prod_{i\in V}(1+\phi(q_i))
			=\sum_{A\subset V}\prod_{i\in A}\phi(q_i),
		\end{equation}
		after choosing
		\begin{equation}\label{comparison}
			1+\phi(q):=\prod_{j=1}^{N_R}(1+f^{ls}(p_j,q)).
		\end{equation}
		Hence, we are also in the context of the abstract polymer model considered in \cite{kuna2018convergence}, formula (3.4),
		for the above choice of the test function $\phi$.
		However, the goal there was to express the $n$-body truncated correlation function in terms of
		a sum over graphs by individuating the corresponding terms when $n$ small particles are fixed
		by the test functions.
		On the contrary, here, we want to compute the effective multi-body potential by fixing the centers
		of some big spheres rather than the small ones like previously.
		These two quantities seem similar, but they have different combinatorics since we have to further expand
		over the product over $j=1,\ldots, N_R$ in \eqref{comp_act}. 
		For that reason we cannot be based directly on the result in \cite{kuna2018convergence}, but instead give a hierarchy of the interactions with the big spheres as in the sequel. Nevertheless, in the proofs, the analysis in the former work will be relevant.
	\end{remark}

	We introduce an alternative equivalent expression for the series in \eqref{calc}.
	Let $\mathbb Y_{\Lambda}:=\oplus_{n\geq 1}\mathbb Q^n$,
	where $\mathbb Q:=\oplus_{k\geq 2}(\mathcal P_k([N_r])\times \Lambda^k)$
	and $\mathcal P_{k}([N_r]):=\{V\subset [N_r]:\, |V|=k\}$.
	On $\mathbb Q$ we introduce the measure
	\begin{equation}\label{MeasureXi}
		\d\xi(q_1,\ldots,q_k)\;:=\sum_{g\in \mathcal{C}_k}\prod_{\{i,j\}\in E(g)}f^{ss}(q_i,q_j)\frac{\d q_1}{|\Lambda|}\ldots \frac{\d q_k}{|\Lambda|},
	\end{equation}
	for an element ${\mathbf q}\in\mathbb Q$ with $|{\mathbf q}|=k$.
	Given a generic element $Y\equiv (\mathbf q_{V_1},\ldots,\mathbf q_{V_n})$, for some $n\in\mathbb N$ and some $V_1,\ldots,V_n\subset [N_r]$, we define the cardinality of $Y$ by $|Y|:=\left| \cup_{i=1}^n V_i\right|$.
	Then on $\mathbb Y_{\Lambda}$ we introduce the measure $\d\nu_{\Lambda,N_r,\mathbf p}$ given by:
	\begin{equation}\begin{split}
			\int_{\mathbb Y_{\Lambda}}h(Y)\d\nu_{\Lambda,N_r,\mathbf p}(Y) & := \sum_{n=1}^{\infty}\frac{1}{n!}\sum_{(V_1,\ldots,V_n)\in\mathcal{V}_{r}^n}\varphi^T(V_1,...,V_n)\left(\frac{|\Lambda|}{|\tilde\Lambda(\mathbf p)|}\right)^{\sum_{i=1}^n |V_i|}\times\\
			&\hspace{.3cm}\int\limits_{\Lambda^{\card{V_1}}}\cdots \int\limits_{\Lambda^{\card{V_n}}} h\bigg((\mathbf q_{V_i})_{i=1}^n\bigg) \prod_{i=1}^n \d\xi(\mathbf q_{V_i}),
			\label{DefYMeas}
	\end{split}\end{equation}
	for some $h\in L^{\infty}(\mathbb Y_{\Lambda},\mathbb R)$.
	
	For the case $h\equiv 1$, we have:
	\begin{equation}\begin{split}\label{only_small_notation}
			\int_{\mathbb Y_{\Lambda}}\d\nu_{\Lambda,N_r,\emptyset}(Y) & = 
			\frac{1}{|\Lambda|} \log Z^{\emptyset}_{\Lambda,N_r}\\
			& = 
			\frac{N_r}{|\Lambda|}\sum_{k\ge 1}\frac{1}{k+1}P_{|\Lambda|,N_r}(k)B^{\emptyset}_{\Lambda}(k),
	\end{split}\end{equation}
	where, following \cite{pulvirenti2012cluster},
	\begin{equation}\label{B_empty}
		B^{\emptyset}_{\Lambda}(k):=
		\frac{|\Lambda|^k}{k!}\sum_{\substack{(V_1,\ldots,V_n)\in\mathcal{V}_{r}^n \\ \cup_{i=1}^nV_i=[k+1]}}\varphi^T(V_1,...,V_n)
		\int\limits_{\Lambda^{\card{V_1}}}\cdots \int\limits_{\Lambda^{\card{V_n}}}\prod_{i=1}^n\d\xi(\mathbf q_{V_i}).
	\end{equation}
	Note that from \cite{pulvirenti2012cluster}, Section 6, we have that
	\begin{equation}\label{est}
		|B^{\emptyset}_{\Lambda}(k)-\beta_{\Lambda}(k)|\leq
		C_k\frac{1}{|\Lambda|},
	\end{equation}
	where
	\begin{equation}\label{betaLambda}
		\beta_\Lambda(k)=\frac{1}{|\Lambda|}\frac{1}{k!}\sum_{g\in\mathcal B_{k+1}}\int_{\Lambda^{k+1}}\prod_{(i,j)\in E(g)}
		f_{i,j}(q_i,q_j)\d\mathbf q_{[k+1]}
	\end{equation}
	and $\lim_{\Lambda\to\mathbb R^d}\beta_{\Lambda}(k)=\beta_k$, as given in \eqref{betan}.
	Furthermore, by the convergence of the cluster expansion and the dominated convergence theorem one
	can pass the limit inside the sum over $k$ in \eqref{only_small_notation}.
	The same is true for a generic bounded $h$.
	With this new notation, equation \eqref{calc} can be written as follows:
	\begin{eqnarray}\label{calc15}
		\frac{1}{|\Lambda|}\ln Z^{\mathbf{p}}_{\Lambda,N_r} & = &
		\frac{N_r}{|\Lambda|}\ln\frac{|\tilde\Lambda(\mathbf{p})|}{|\Lambda|}
		+\frac{1}{|\Lambda|}
		\int_{\mathbb Y_{\Lambda}}
		(1+\Theta_{\mathbf p}(Y))
		\d\nu_{\Lambda,N_r,\mathbf p}(Y),
	\end{eqnarray}
	where, for a generic element $Y\equiv (\mathbf q_{V_1},\ldots,\mathbf q_{V_n})$, we define:
	\begin{equation}\label{key_formula}
		1+\Theta_{\mathbf p}(Y):=
		\prod_{i=1}^n (1+\vartheta_{\mathbf p}(\mathbf q_{V_i}))=\prod_{j=1}^{N_R}(1+\zeta(p_j,Y))=1+\sum_{\substack{J\subset [N_R] \\ J\neq\emptyset}}\prod_{j\in J}\zeta(p_j,Y),
	\end{equation}
	and for some $p\in\Lambda$:
	\begin{equation}\begin{split}\label{zeta}
			-1\leq 
			\zeta(p,Y) & :=  \prod_{i=1}^n\prod_{k\in V_i}(1+f^{ls}(p,q_k))-1\\
			& =: 
			\prod_{i=1}^n (1+\bar\theta(p,\mathbf q_{V_i}))-1\leq 0.
	\end{split}\end{equation}
	Note that in $\zeta(p,Y)$ there is at least one link from $p$ to some $q_k$ in $V_i$, for some $i=1,\ldots,n$. 
	The re-arrangement in \eqref{key_formula} 
	will be particularly useful in the next step of integrating over the large spheres
	where we need to organize the sum
	in terms of the indices of the large spheres involved.
	
	With this observation, equation \eqref{calc15} can be rewritten as follows:
	\begin{equation}\begin{split}\label{calc2}
			\frac{1}{|\Lambda|}\ln Z^{\mathbf{p}}_{\Lambda,N_r} & = 
			\frac{N_r}{|\Lambda|}\ln\frac{|\tilde\Lambda(\mathbf{p})|}{|\Lambda|}
			+\frac{1}{|\Lambda|}
			\int_{\mathbb Y_{\Lambda}}\d\nu_{\Lambda,N_r,\mathbf p}(Y)+
			\\
			&
			+\sum_{\substack{J\subset[N_R]\\ J\neq\emptyset}}\frac{N_r}{|\Lambda|}\sum_{k\ge1}
			\frac{1}{k+1}P_{|\Lambda|,N_r}(k)B^{\mathbf{p}}_{\Lambda}(k,J),
	\end{split}		\end{equation}
	where $P_{|\Lambda|,N_r}$ is given in \eqref{P} and we introduced the
	following quantity:
	\begin{equation}\begin{split}\label{Bp}
			B^{\mathbf{p}}_{\Lambda}(k,J) & :=  \frac{|\Lambda|^k}{k!}
			\sum_{n\ge1}\frac{1}{n!}\sum_{\substack{(V_1,...,V_n)\in\mathcal{V}^n_{r}\\ \bigcup_{i=1}^{n}V_i=[k+1]}}\varphi^T(V_1,...,V_n)
			\left(\frac{|\Lambda|}{|\tilde\Lambda(\mathbf p)|}\right)^{\sum_{i=1}^n|V_i|}\times\\
			&\hspace{2.9cm}
			\int\limits_{\Lambda^{\card{V_1}}}\cdots \int\limits_{\Lambda^{\card{V_n}}}\prod_{j\in J}\zeta(p_j,(\mathbf q_{V_i})_{i=1}^n)\prod_{i=1}^n \d\xi(\mathbf q_{V_i}).
	\end{split}\end{equation}
	Note that following \cite{kuna2018convergence} it has similar properties with $B^{\emptyset}_{\Lambda}$ as it will be discussed later in detail in Lemma~\ref{second_cancellation2}.
	
	\begin{remark}\label{amb}
		Suppose that in the integration in \eqref{Bp}
		the center of a small sphere has label $k\in [N_r]$ with
		$k\in V_i\cap V_{i'}$, for some $1\leq i<i'\leq n$.
		In such a scenario the same function $f^{ls}(p,q_k)$ will be integrated in both $\d\xi(\mathbf q_{V_i})$ and $\d\xi(\mathbf q_{V_{i'}})$.
		In order to properly perform these integrations we
		introduce the notation of $q_k^i$ with an upper index $i$ indicating the polymer $V_i$ it belongs.
		We give a simple example: 
		say $V_1=V_2=\{1,2\}$;
		then we have the following integration:
		\begin{equation}\begin{split}\label{exint}
				&\int\limits_{\Lambda^{\card{V_1}}}\int\limits_{\Lambda^{\card{V_2}}}f^{ls}(p_1,q^1_1)f^{ls}(p_1,q^2_1)\prod_{i=1,2}\d\xi((q^i_k)_{k\in V_i})\\
				&\hspace{0.5cm}= 
				\int_{\Lambda^4}f^{ls}(p_1,q^1_1)f^{ls}(p_1,q^2_1)f^{ss}(q^1_1,q^1_2)f^{ss}(q^2_1,q^2_2)\frac{\d q^1_1}{|\Lambda|}
				\frac{\d q^1_2}{|\Lambda|}\frac{\d q^2_1}{|\Lambda|}\frac{\d q^2_2}{|\Lambda|}\\
				&\hspace{0.5cm}= 
				\left(
				\int_{\Lambda^2}f^{ls}(p_1,q^1_1)f^{ss}(q^1_1,q^1_2)
				\frac{\d q^1_1}{|\Lambda|}\frac{\d q^1_2}{|\Lambda|}
				\right)^2,
		\end{split}\end{equation}
		i.e., the functions $q\mapsto f^{ls}(p_1,q)$ should appear in two different integrals because they belong to two different polymers $V_i$, $i=1,2$. Thus, in the notation $q^i_k$, the lower index indicates the actual label of the particle chosen within a polymer $V$, 
		while the upper index indicates in which integral the actual particle should be considered,
		in order to properly perform the integrations. So in the above example, we have two particles, but considered twice in
		the integration while calculating their weight in \eqref{DefYMeas}.
	\end{remark}

	As far as the second term of \eqref{calc2} is concerned, this is almost equal to \eqref{only_small_notation}
	up to a different normalization (taking into account the presence of the big spheres).
	In fact, their difference is given by:
	\begin{equation}\begin{split}\label{discr}
			&
			\int_{\mathbb Y_{\Lambda}}\d\nu_{\Lambda,N_r,\mathbf p}(Y)
			- \int_{\mathbb Y_{\Lambda}}\d\nu_{\Lambda,N_r,\emptyset}(Y)=
			N_r\sum_{k\geq 1}\frac{1}{k+1}P_{N_r,|\Lambda|}(k)\;\frac{|\Lambda|^k}{k!}\times
			\\
			&
			\sum_{n\geq 1}\sum_{\substack{(V_1,\ldots,V_n)\in\mathcal{V}_{r}^n \\ \cup_{i=1}^n=[k+1]}}\varphi^T(V_1,...,V_n)
			\left[
			\left(\frac{|\Lambda|}{|\tilde\Lambda(\mathbf p)|}\right)^{\sum_{l=1}^n |V_i|}-1
			\right]
			\int\limits_{\Lambda^{\card{V_1}}}\cdots \int\limits_{\Lambda^{\card{V_n}}}\prod_{i=1}^n\d\xi(\mathbf q_{V_i}).
	\end{split}\end{equation}
	Using \eqref{ratio} this can be bounded from above by
	\begin{equation}\label{Aub}
		N_r\sum_{k\geq 1}\frac{1}{k+1}P_{N_r,|\Lambda|}(k)A_{\Lambda}\left(k;\frac{N_R |B_{R+r}|}{|\Lambda|}\right)
	\end{equation}
	and from below by
	\begin{equation}\label{Alb}
		N_r\sum_{k\geq 1}\frac{1}{k+1}P_{N_r,|\Lambda|}(k)A_{\Lambda}\left(k;\frac{N_R |B_R|}{|\Lambda|}\right),
	\end{equation}
	where for some $\rho\in\mathbb R$
	\begin{equation}\begin{split}\label{ALambda}
			A_{\Lambda}(k;\rho):=&
			\frac{|\Lambda|^k}{k!}\times\\
			&\sum_{n\geq 1}\sum_{\substack{(V_1,\ldots,V_n)\in\mathcal{V}_{r}^n \\ \cup_{i=1}^n V_i=[k+1]}}\varphi^T(V_1,...,V_n)
			\left[
			\left(\frac{1}{1-\rho}\right)^{\sum_{i=1}^n |V_i|}-1
			\right]
			\int\limits_{\Lambda^{\card{V_1}}}\cdots \int\limits_{\Lambda^{\card{V_n}}}\prod_{i=1}^n\d\xi(\mathbf q_{V_i}).
	\end{split}\end{equation}
	We recall that we were obliged to introduce this discrepancy in order to properly organize the cluster
	expansion of the canonical ensemble in \eqref{att}.

	Collecting all above terms we obtain a new expression of the partition function 
	$\hat{Z}_{\Lambda,N_R, N_r}$
	in terms of the effective interactions between the big spheres. 
	It is summarized in the following lemma:

	\begin{lemma}\label{Effective1}
		Under condition \eqref{cond1} 
		the partition function \eqref{EffectiveCPF} can be expressed as follows:
		\begin{equation}\begin{split}\label{eff}
				\hat{Z}_{\Lambda,N_R, N_r} &= 
				\int\limits_{\Lambda^{N_R}}
				\exp\left[
				\int_{\mathbb Y_{\Lambda}}\d\nu_{\Lambda,N_r,\mathbf p}(Y)
				-
				\int_{\mathbb Y_{\Lambda}}\d\nu_{\Lambda,N_r,\emptyset}(Y)		\right]
				\times\\
				&\prod_{1\le j<j'\le N_R}[1+f^{ll}(p_j,p_{j'})] e^{W^{(1)}(\mathbf p)}
				\exp\left[-\sum_{J\subset [N_R]\atop \card{J}\ge 2}W_{|J|}(\mathbf p_J)
				\right]
				\prod_{i=1}^{N_R}\frac{\mu_{\Lambda,N_r}(\d p_i)}{|\Lambda|},
		\end{split}\end{equation}
		where we have the following contributions: the multi-body interaction via clouds $Y\in\mathbb Y_\Lambda$
		\begin{eqnarray}
			W_{\card{J}}(\mathbf p_J) & := & -\int\limits_{\mathbb Y_{\Lambda}}\prod_{j\in J}\zeta(p_j,Y) \d\nu_{\Lambda,N_r,\mathbf p}(Y),\label{W}
		\end{eqnarray}
		multi-body interactions via clouds consisting of single small particles
		\begin{eqnarray}
			W^{(1)}(\mathbf p) & := &
			\ln\left(\frac{|\tilde\Lambda(\mathbf p)|}{|\Lambda|}\right)^{N_r} \label{W1}
		\end{eqnarray}
		and the new prior measure of big spheres interacting with a cloud
		\begin{equation}\label{mus}
			\mu_{\Lambda, N_r, \mathbf p}(\d p):=\exp\left[\int\limits_{\mathbb Y_{\Lambda}}\zeta(p,Y)\d\nu_{\Lambda,N_r,\mathbf p}(Y)\right]\d p.
		\end{equation}
		All these terms are well-defined, in particular,
		\begin{equation}\label{wdW}
			\int_{\mathbb Y_{\Lambda}}\prod_{j\in J}|\zeta(p_j,Y)| \,\d|\nu|_{\Lambda,N_r,\mathbf p} (Y)<\infty,
		\end{equation}
		for all $\mathbf p$.
	\end{lemma}	
	
	\begin{proof}
		Starting from \eqref{EffectiveCPF} we compute the partition function 
		$Z^{\mathbf{p}}_{\Lambda,N_r}$ of the small spheres in the presence of the big ones as given in \eqref{Zps}.
		This is given in Lemma~\ref{PropXuan2}
		and in particular in formula \eqref{ClusterZps2}.
		Then, from \eqref{ClusterZps3} we have that \eqref{W} and \eqref{mus} are well-defined.
		For the first term of \eqref{calc2} we use the bound \eqref{ratio}.
	\end{proof}

	\section{Effective canonical partition function for the big spheres}\label{s_big}
	
	In the previous section, integrating over the small spheres we obtained \eqref{eff} which is a canonical 
	partition function for a fixed number $N_R$
	of big spheres interacting via two different types of potentials: the pair potential $V^{ll}$ and the multi-body effective potential $W_{|J|}$ given in \eqref{W}.
	Furthermore, the
	integration is with respect to the new reference measure $\mu_{\Lambda,N_r}$, given in \eqref{mus}.
	In this section we compute the logarithm of this effective canonical partition function
	using two main tools:
	first the cluster expansion strategy for the canonical ensemble as it was suggested 
	in \cite{pulvirenti2012cluster} and second the tree-graph inequality obtained in \cite{jansen2020cluster}. 
	Note that \cite{pulvirenti2012cluster} does not cover the case of multi-body potentials. However, this is not a problem here since we can exploit the fact that the particular
	multi-body potentials we obtain are the result of pair interactions between the small spheres.
	This was the case in \cite{jansen2020cluster} as well.
	
	We expand all factors in the partition function
	\eqref{eff}.
	First, following \cite{jansen2020cluster} for the multi-body potential we have:
	\begin{equation}\label{ExpSvil}
		\exp\left\{-\sum_{\substack{J\subset[N_R],\\|J|\ge2}}W_{|J|}(\mathbf{p}_J)\right\}= \sum_{k\ge0}\frac{1}{k!}\int_{\mathbb{Y}^k_{\Lambda}}
		\sum_{\gamma}^{N_R,k}	
		\prod_{\{i,j\}\in E(\gamma)}\zeta(p_i,Y_j)\prod_{i=1}^{k}\nu_{\Lambda,N_r,\mathbf p}(dY_{N_R+i}).
	\end{equation}
	The graphs $\gamma$ in the second sum in \eqref{ExpSvil} have set of vertices in $\{1,...,N_R,N_R+1,...,N_R+k\}$ and are such that (i) for all $\{i,j\}\in E(\gamma),\;i\le N_R,\;j\ge N_R+1$, (ii) if $\{i_1,j\}\in E(\gamma)$ then there exists $i_2\le N_R$ such that $\{i_2,j\}\in E(\gamma)$, where $i_1,i_2\leq N_R$ and $j> N_R$.
	Second, the one-body contribution can be written as:
	\begin{equation}\begin{split}\label{extra}
			e^{W^{(1)}(\mathbf p)}
			&=
			\left(1+\int_{\Lambda}\left[\prod_{j=1}^{N_R}(1+f^{ls}(p_j,q))-1\right]\frac{\d q}{|\Lambda|}\right)^{N_r}\\
			&=
			\sum_{A\subset [N_r]}
			\prod_{i\in A}
			\left(
			\int_{\Lambda}\sum_{\substack{J_i\subset [N_R] \\ J_i\neq\emptyset}}\prod_{j\in J_i}f^{ls}(p_j,q)\frac{\d q}{|\Lambda|}
			\right)\\
			&=
			\sum_{A\subset [N_r]}
			\sum_{\substack{J_1,\ldots,J_{|A|} \\ J_i\subset [N_R], J_i \neq\emptyset}}
			\prod_{i=1}^{|A|}
			\int_{\Lambda}\prod_{j\in J_i}f^{ls}(p_j,q)\frac{\d q}{|\Lambda|}.
	\end{split}\end{equation}
	Multiplying \eqref{extra} with \eqref{ExpSvil}
	we obtain the same formula but on the augmented space
	\begin{equation}\label{augY}
		\tilde{\mathbb{Y}}:=\mathbb{Y}\oplus \Lambda,
	\end{equation}
	i.e., we include the case that a cloud consists of only one particle as in \eqref{extra}.
	Furthermore, we 
	get extra terms of interactions of small spheres with only one big sphere, so we
	should also augment the space in the integration in the exponent of the measure $\mu_{\Lambda,N_r}$.
	Furthermore, we note that in the case of only one particle, i.e., Y=\{q\} we define:
	\begin{equation}\label{linkwithone}
		\zeta(p,Y):=f^{ls}(p,q).
	\end{equation}

	Next, we expand the product $\displaystyle\prod_{1\le i<j\le N_R}[1+f^{ll}(p_i,p_j)]$.
	In order to organize a proper bookkeeping of all terms we consider the following set of polymers:
	\begin{equation}\label{VR}
		\mathcal{V}_R:=\{V\subset[N_R]:\;|V|\ge2\}.
	\end{equation}
	Given $V,V'\in\mathcal{V}_R$, we say that $V$ is compatible with $V'$ ($V\sim V'$) if and only if $V\cap V'=\emptyset$.

	Overall we obtain graphs with two types of links so we
	recall the following definition  from \cite{jansen2020cluster}.
	
	\begin{definition}\label{DefHyperGraphs}
		For $m,r\in\mathbb{N}_0$ with $m+r\ge 1$, let $\mathcal{G}^*_{m,r}$ be the class of graphs with vertex set $\{1,...,m+r\}$ such that:
		\begin{enumerate}
			\item The graph has no edges $\{k_1,k_2\}$ with $k_1,k_2\ge m+1$.
			\item Every vertex $k\in\{m+1,...,m+r\}$ belongs to at least two distinct edges $\{s,k\},\;\{t,k\},\;t,s\in\{1,...,m\}$.
		\end{enumerate}	
	\end{definition}
	
	Furthermore, we let
	\begin{equation}\label{C*} 
		\mathcal{C}^*_{m,k}:=\mathcal{G}^*_{m,k}\cap\mathcal{C}_{m+k},
	\end{equation}
	where $\mathcal{G}^*_{m,k}$ is given by Definition \ref{DefHyperGraphs} and $\mathcal{C}_{m+k}$ is the set of connected graphs with set of vertices $\{1,\ldots,m, m+1,..., m+k\}$. 
	
	We aim at representing the partition function \eqref{eff} by an abstract polymer model.
	Apart from the terms we treated above we have another $\mathbf p$-depended contribution
	(the first line in \eqref{eff}) 
	which instead of considering it in the cluster expansion we will simply give upper and
	lower bounds, as we have already commented in Section~\ref{s_small}.
	At this point we are ready to introduce the ``activity'' for a polymer $V\in\mathcal V_R$:
	\begin{equation}\begin{split}\label{w}
			&w_{\Lambda}(V)\\
			&\hspace{0.2cm}:=\sum_{k\ge 0}\frac{1}{k!}\int_{\Lambda^{|V|}}\int_{\mathbb{Y}^k}\varphi^T_*(\underline p_{V};Y_{|V|+1},...,Y_{|V|+k})
			\prod_{j\in V}\frac{\mu_{\Lambda,N_r,\mathbf p}(\d p_j)}{\int_{\Lambda}\mu_{\Lambda,N_r,\mathbf p}(\d p)}\prod_{i=1}^k\nu_{\Lambda,N_r,\mathbf p}(dY_{|V|+i}),\\
			&\hspace{0.2cm}=\sum_{k\ge 0}\frac{1}{k!}\int_{\Lambda^{|V|}}\int_{\mathbb{Y}^k}\varphi^T_*(\underline p_{V};Y_{|V|+1},...,Y_{|V|+k})
			\prod_{j\in V}\frac{\d p_j}{|\Lambda|}\prod_{i=1}^k\nu_{\Lambda,N_r,\mathbf p}(dY_{|V|+i}),
	\end{split}\end{equation}
	where we used the fact that $\int_{\mathbb Y_{\Lambda}}\zeta(p,Y) \d\nu_{\Lambda,N_r,\mathbf p}(Y)$ is independent of $p$ (see also Lemma~\ref{second_cancellation1}
	for a detailed analysis of this term) and hence we have:
	\begin{equation}\label{vol}
		\mu_{\Lambda,N_r}(\d p)= C_{\Lambda, N_r,\mathbf p} \d p, \quad C_{\Lambda, N_r,{\bf{p}}}:= e^{\int_{\mathbb Y_{\Lambda}}\zeta(0,Y) \d\nu_{\Lambda,N_r,\mathbf p}(Y)},
	\end{equation}
	due to periodic boundary conditions.
	Furthermore, following \cite{jansen2020cluster}, we defined
	\begin{equation}\label{phistar}
		\varphi^T_*(p_1,...,p_n;Y_{n+1},...,Y_{n+k}):=\sum_{g\in{
				\mathcal{C}^*_{n,k}}}\prod_{\substack{\{i,j\}\in E(g)\\i,j\in V}}f^{ll}(p_i,p_j)\prod_{\substack{\{i,j\}\in E(g)\\i,\in V,\;j\ge |V|+1}}\zeta(p_i,Y_j),
	\end{equation}
	see also Figure~\ref{fig:1}.
	\begin{figure}[h]
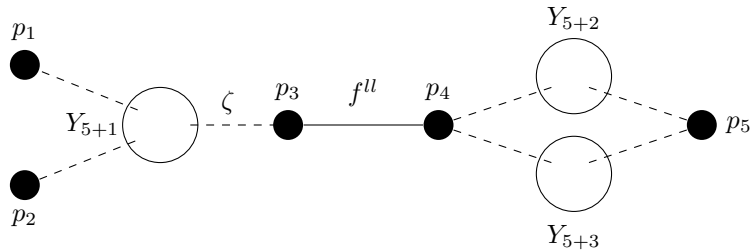

		\begin{center}
			\begin{tabular}{rcl}
				\tikz[baseline=-2pt]{ 
					\draw[dashed](-1.5,-0.8)node[bigvertex, label=below:{$p_2$}]{}--(0,-0.2)node[]{};
					\draw[dashed](0,0.2)--(-1.5,0.8)node[bigvertex, label=above:{$p_1$}]{};
					\draw(0.3,0)circle(0.5)[]{}node[label=left:{$Y_{5+1}$\;\;\;}]{};
					\draw[dashed](0.7,0)--(2,0)node[bigvertex,label=above:{$p_3$}](v1){};
					\draw(1.2,-0.1)node[label={$\zeta$}]{};
					\draw[dashed](5.5,0.5)--(4,0)node[bigvertex,label=above:{$p_4$}](v2){};
					\draw[dashed](5.5,-0.5)--(v2);
					\draw(v1)--(v2);
					\draw(3,0)node[label={$f^{ll}$}]{};
					\draw(5.8,0.65)circle(0.5);
					\draw(5.8,1)node[label={$Y_{5+2}$}]{};
					\draw(5.8,-1.9)node[label={$Y_{5+3}$}]{};
					\draw(5.8,-0.65)circle(0.5);
					\draw[dashed](6,0.5)--(7.5,0)node[bigvertex,label=right:{$p_5$}](v3){};
					\draw[dashed](6,-0.5)--(v3);
					
				}
			\end{tabular}
		\end{center}\caption{A graph in $\mathcal{C}^*_{5,3}$, contributing to $\varphi_*^T(p_1,p_2,p_3,p_4,p_5;Y_6,Y_7,Y_8)$.}
		\label{fig:1}
	\end{figure}	
	
	This is equivalent to \cite{jansen2020cluster} but in the context of the canonical ensemble:
	we choose $n$ big spheres among a total of $N_R$ many and an arbitrary number $k\geq 0$ of blobs $Y_{n+1},\ldots,Y_{n+k}$ that come from the effective interaction through the small
	spheres.
	We recall that blobs can have one or more particles, as discussed in \eqref{augY}. In the case of only one particle the link is given by \eqref{linkwithone}.
	We have the following proposition:

	\begin{prop}\label{ClusterEffective2}
		Let $a>0$ and $b,c:\mathbb Y_{\Lambda}\to\mathbb R$ with $b(Y):=b|Y|$, for $b>0$, similarly for $c$.
		Under the assumptions of Lemma~\ref{Effective1} as well as
		\begin{equation}\label{c1}
			e^{b+c}|B_{R+r}(0)\setminus B_{R-r}(0)|\frac{N_r}{|\Lambda|-N_R |B_{R+r}|} +e^a|B_{2R}|\frac{N_R}{|\Lambda|}\le a
		\end{equation}
		and
		\begin{equation}\label{c2}
			e^a|B_{R+r}(0)\setminus B_{R-r}(0)|\frac{N_R}{|\Lambda|}\le b,
		\end{equation}
		for some $a,b,c\ge0$, we have the following convergent upper bound (similarly for the lower bound):
		\begin{equation}\begin{split}\label{ClusterEff}
				\frac{1}{|\Lambda|}\log \hat{Z}^u_{\Lambda,N_R,N_r} & := 
				\frac{N_r}{|\Lambda|}\sum_{k\geq 1}\frac{1}{k+1}
				P_{N_r,|\Lambda|}(k)A_{\Lambda}(k;\frac{N_{R}|B_{R+r}|}{|\Lambda|})+\\
				&
				+\frac{N_R}{|\Lambda|}\log\left(\int_{\Lambda}\frac{\mu_{\Lambda,N_r}(\d p)}{|\Lambda|}\right)\\
				&
				+\frac{1}{|\Lambda|}\sum_{n\ge1}\frac{1}{n!}\sum_{(V_1,...,V_n)\in\mathcal{V}^n_R}\varphi^T(V_1,...,V_n)\prod_{i=1}^n w_{\Lambda}(V_i),
		\end{split}\end{equation}	
		where all infinite series in the right hand side are absolutely convergent. 
		More precisely, we
		have the following estimates.
		For a fixed big sphere with center $p_1\in\Lambda$:
		\begin{equation}\begin{split}\label{res1}
				&\sum_{n\ge 2}\frac{1}{(n-1)!}e^{a n}\left(\frac{N_R}{|\Lambda|}\right)^{n-1}
				\sum_{k\ge 0}\frac{1}{k!}
				\int_{\Lambda^n}\int_{\mathbb{Y}^k}
				|\varphi^T_*(p_1,...,p_n;Y_{n+1},...,Y_{n+k})|\\
				&\hspace{5.3cm}
				\prod_{i=1}^{k}|\nu|_{\Lambda,N_r,\mathbf p}(dY_{n+i})\prod_{i=2}^{n}\d p_i\leq
				e^{a}-1.
		\end{split}\end{equation}
		For a fixed cloud $Y_{n+1}\in\mathbb Y_{\Lambda}$:
		\begin{equation}\begin{split}\label{res2}
				&\sum_{n\ge 2}\frac{1}{n!}e^{a n}\left(\frac{N_R}{|\Lambda|}\right)^{n-1}
				\sum_{k\ge 1}\frac{1}{(k-1)!}
				\int_{\Lambda^n}\int_{\mathbb{Y}^k}|\varphi^T_*(p_1,...,p_n;Y_{n+1},...,Y_{n+k})|\\
				&\hspace{4.5cm}
				\prod_{i=1}^{k}|\nu|_{\Lambda,N_r,\mathbf p}(dY_{n+i})\prod_{i=2}^{n}\d p_i\leq
				e^{b(Y_{n+1})}-1.
		\end{split}\end{equation}
	\end{prop}
	
	\begin{proof}
		For the first factor of \eqref{eff} we use the bounds
		\eqref{Aub} and \eqref{Alb}.
		Note that in this fashion the upper and lower bounds of the first factor do not depend 
		on the configuration $\mathbf p$ anymore.
		Hence, by expanding the other three factors as described above
		as well as normalizing the prior measure we obtain:
		\begin{equation}\begin{split}
				\left(\frac{|\Lambda|}{\int_{\Lambda}\mu_{\Lambda,N_r}(\d p)}\right)^{N_R}\hat{Z}_{\Lambda,N_R,N_r}
				&\leq
				\exp\left[N_r\sum_{k\geq 1}\frac{1}{k+1}
				P_{N_r,|\Lambda|}(k)A_{\Lambda}(k;\frac{N_R |B_{R+r}|}{|\Lambda|})\right]
				\times\\
				&\hspace{0.7cm}\left(
				1+\sum_{n\geq 1}\frac{1}{n!}
				\sum_{\substack{(V_1,\ldots,V_n)\\ V_i\in\mathcal V_R}}
				\varphi(V_1,...,V_n)
				\prod_{i=1}^n w_{\Lambda}(V_i)
				\right)
		\end{split}\end{equation}
		and similarly for the lower bound with $A_{\Lambda}(k;\frac{N_R |B_R|}{|\Lambda|})$, respectively.
		The activity $w_{\Lambda}$ is defined in \eqref{w}. Having the bounds of the 
		partition function
		in the form of the abstract polymer model we check the
		convergence condition as in \cite[ Theorem 2.1]{poghosyan2009abstract}.
		Given $a>0$ we have:
		\begin{equation}\begin{split}\label{ConvSecondStepCanonical}
				&\sup_{i\in[N_R]}\;\sum_{\substack{V\in\mathcal{V}_R:\\ V\ni i}}|w_{\Lambda}(V)|e^{a|V|}\\
				&\hspace{.3cm}\leq \sum_{n\ge 2}e^{a n}\binom{N_R-1}{n-1}\times\\
				&\hspace{.7cm}\sum_{k\ge 0}\frac{1}{k!}
				\int_{\Lambda^n}\int_{\mathbb{Y}^k}|\varphi^T_*(p_1,...,p_n;Y_{n+1},...,Y_{n+k})|\prod_{i=1}^{k}|\nu|_{\Lambda,N_r,\mathbf p}(dY_{n+i})\prod_{i=2}^{n}\frac{\d p_i}{|\Lambda|}\\
				& \hspace{.3cm}= e^a\sum_{n\ge 2}\frac{1}{(n-1)!}\left(e^a\frac{N_R}{|\Lambda|}\right)^{n-1}\times\\
				&\hspace{.9cm}\int_{\Lambda^n}
				\sum_{k\ge 0}\frac{1}{k!}
				\int_{\mathbb{Y}^k}|\varphi^T_*(p_1,...,p_n;Y_{n+1},...,Y_{n+k})|\prod_{i=1}^{k}|\nu|_{\Lambda,N_r,\mathbf p}(dY_{n+i})\prod_{i=2}^{n}\d p_i.
		\end{split}\end{equation}
		
		Now we proceed as in \cite{jansen2020cluster}, in the proof of Theorem~3.8, by applying the
		tree-graph inequality to $\varphi^T_*(p_1,...,p_n;Y_{n+1},...,Y_{n+k})$.
		Then, we just need to integrate over the various leaves of the trees: big-big and big-cloud.
		In the first case, a vertex corresponding to a big sphere comes with a factor $|B_{2R}|$
		(excluded volume between two big) 
		while in the second with a factor $|B_{R+r}|$ (excluded volume of a big and a small).
		In both cases, we observe that we can incorporate it with  $\displaystyle\frac{N_R}{|\Lambda|}$
		and consider $\displaystyle\frac{N_R |B_{R+r}|}{|\Lambda|}$
		as the macroscopic quantity of the density of the big spheres.
		This fact is not really needed here, but we mention it as it will be discussed in Remark~\ref{second_remark}.
		
		In order to conclude we need
		to check that under \eqref{c1} and \eqref{c2}
		the sufficient conditions of Theorem 3.8 from \cite{jansen2020cluster} are met.
		First, recalling the notation in Remark~\ref{amb}, for $Y=((q^i_k)_{k\in V_i})_{i=1}^{n}$ we call
		\begin{equation}
			\tilde{\zeta}(p,Y):=-\mathbf{1}_{\{\exists i\in\{1,...,n\},\;\exists\; k_0\in V_i\;|\;  R-r<|p-q^{i}_{k_0}|< R+r\}}.
		\end{equation}
		Note that Assumption 2 of \cite{jansen2020cluster} is satisfied, i.e., we have: 
		\begin{equation}
			\prod_{1\le  <j\le N_R}[1+f^{ll}(p_i,p_j)]\prod_{i=1}^{N_R}|\zeta (p_i,Y)|\le\prod_{1\le i <j\le N_R}[1+f^{ll}(p_i,p_j)]\prod_{i=1}^{N_R}|\tilde{\zeta} (p_i,Y)|,
		\end{equation}
		\begin{equation}
			|\zeta (p_1,Y)|\prod_{1\le i <j\le N_R}[1+f^{ll}(p_i,p_j)]\left|\prod_{i=2}^{N_R}\zeta (p_i,Y)-1\right|\le |\tilde{\zeta} (p_1,Y)|\prod_{1\le i <j\le N_R}[1+f^{ll}(p_i,p_j)].
		\end{equation}
		Then, 
		recalling \eqref{DefYMeas}, for fixed $p\in\Lambda$ the ``big-cloud'' leaf yields:
		\begin{equation}\begin{split}\label{one}
			&\int_{\mathbb Y}|\tilde\zeta(p,Y)| e^{b(Y)} \d |\nu|_{\Lambda,N_r,\mathbf p}(Y) \\
			& \hspace{.5cm}= 
			\sum_{n=1}^{\infty}\frac{1}{n!}\sum_{(V_1,\ldots,V_n)\in\mathcal{V}_{r}^n}|\varphi^T(V_1,...,V_n)|
			e^{\sum_{i=1}^n b(V_i)}\times\\
			&\hspace{1.8cm}\int\limits_{\Lambda^{\card{V_1}}}\cdots \int\limits_{\Lambda^{\card{V_n}}}\mathbf 1_{\cup_{i=1}^n\cup_{k\in V_i}\{R-r<|p-q^{i}_k|< R+r\}}\prod_{i=1}^n \d\xi\big((q_k^i)_{k\in V_i}\big)\\
			& \hspace{.5cm}= 
			N_r
			\sum_{V_1\subset [N_r], V_1\ni 1}
			e^{b(V_1)}\left(\frac{|\Lambda|}{|\tilde\Lambda(\mathbf p)|}\right)^{|V_1|}\times\\
			&\hspace{1.9cm}\int_{\Lambda^{|V_1|}} \mathbf 1_{\{R-r<|p-q_1|< R+r\}}\sum_{g\in \mathcal{C}_{|V_1|}}\prod_{\{i,j\}\in E(g)}f^{ss}(q_i,q_j)\frac{\d q_1}{|\Lambda|}\ldots \frac{\d q_{|V_1|}}{|\Lambda|}\times\\
			&
			\hspace{1.9cm}\sum_{n=2}^{\infty} \frac{1}{(n-1)!} 
			\sum_{\substack{(V_2,\ldots,V_n)\\ V_i\in\mathcal V_r,\, |V_i|\geq 2}}				|\varphi^T(V_1,...,V_n)|e^{\sum_{k=2}^n b(V_k)}\prod_{k=2}^n |\zeta_{\Lambda}(V_k)|\\
			& \hspace{.5cm}\leq  
			e^c N_r
			\sum_{m=2}^{N_r}\binom{N_r-1}{m-1}
			e^{b m} m^{m-2} \frac{|B_{2r}|^{m-1}}{|\tilde\Lambda(\mathbf p)|^{m-1}} 
			\int_{\Lambda}\mathbf 1_{\{R-r<|p-q_1|< R+r\}}\frac{\d q_1}{|\tilde\Lambda(\mathbf p)|}\\
			& \hspace{.5cm}\leq 
			e^c \frac{N_r}{|\Lambda|-N_R |B_{R+r}|} e^b |B_{R+r}(p)\setminus B_{R-r}(p)| \sum_{m=2}^{N_r}
			\frac{m^{m-2}}{(m-1)!}\left(
			\frac{N_r |B_{2r}|}{|\Lambda|-N_R |B_{R+r}|} e^b
			\right)^{m-1}\\
			& \hspace{.5cm}\leq 
			e^c \frac{N_r}{|\Lambda|-N_R |B_{R+r}|} e^b |B_{R+r}(0)\setminus B_{R-r}(0)|,\end{split}
		\end{equation}
		where for the bound $e^c$ we used \eqref{ClusterZps2}
		while in the last inequality from \eqref{cond1} we computed:
		\begin{equation}
			\sum_{m\ge 2}
			\frac{m^{m-2}}{(m-1)!}\left(
			\frac{N_r |B_{2r}|}{|\Lambda|-N_R |B_{R+r}|} e^b
			\right)^{m-1}\leq \sum_{m\ge 2}
			\frac{m^{m-2}}{(m-1)!}\left(c e^{-c}
			\right)^{m-1}\leq 1.
		\end{equation}
		Furthermore, for the ``big-big'' leaves we have:
		\begin{eqnarray}\label{two}
			\int_{\Lambda}|f^{ll}(p,p')| e^{a}\frac{N_R}{|\Lambda|} \d p'
			\leq 
			e^a\frac{N_R}{|\Lambda|}
			\int_{\Lambda}\mathbf 1_{|p-p'|<2R}
			\d p'= 
			e^a\frac{N_R}{|\Lambda|} |B_{2R}|.
		\end{eqnarray}
		Similarly, for a fixed cloud $Y\in\mathbb Y_{\Lambda}$ we have:
		\begin{eqnarray}\label{three}
			\int_{\Lambda}|\tilde\zeta(p,Y)| e^{a} \frac{N_R}{|\Lambda|}\d p
			& \leq &
			\sup_{q\in Y} e^a\frac{N_R}{|\Lambda|}
			\int_{\Lambda}\mathbf 1_{\cup_{q\in Y}\{R-r<|p-q|< R+r\}}\d p\nonumber\\
			& \leq & 
			|Y| e^a\frac{N_R}{|\Lambda|} |B_{R+r}(0)\setminus B_{R-r}(0)|.
		\end{eqnarray}
		Thus, overall, under the above convergence conditions, both upper and lower limits
		of the partition functions are convergent at finite volume, uniformly on $\Lambda$, 
		$N_r$ and $N_R$. 
		The bounds \eqref{res1} and \eqref{res2} follow from standard arguments.
	\end{proof}

	\begin{remark}\label{key_remark}
		For fixed $r$ and $R$ the convergence conditions \eqref{c1} and \eqref{c2}
		imply that the thermodynamic limit of the free energy exists and
		can be written in terms of $\rho_r:=\lim_{\substack{N_r\to\infty \\ |\Lambda|\to\infty}}\frac{N_r}{|\Lambda|}$
		and $\rho_R:=\lim_{\substack{N_R\to\infty \\ |\Lambda|\to\infty}}\frac{N_R}{|\Lambda|}$.
		It is expected that as in \cite{jansen2020cluster} by integrating first over the small and then
		over the big we obtain an improved radius of convergence.
		Note, however, this improvement was revealed in a further limit of $R/r\to\infty$.
		A similar fact is valid also in the present case of the canonical ensemble.
		Suppose that \eqref{c1} is true for some $|B_{2R}|\frac{N_R}{|\Lambda|}\neq 0$.
		Then, necessarily $\displaystyle a>e^{b+c}|B_{R+r}(0)\setminus B_{R-r}(0)|\frac{N_r}{|\Lambda|-N_R|B_{R+r}|}$,
		where $a$ depends on $R$.
		Recall also that from \eqref{cond1} we have that $\frac{N_r}{|\Lambda|-N_R|B_{R+r}|}$ is of order one.
		Thus, we choose $\displaystyle a:=\alpha |B_{R+r}(0)\setminus B_{R-r}(0)|\frac{N_r}{|\Lambda|-N_R|B_{R+r}|}$
		with $\alpha> e^{b+c}$.
		In this way we obtain that conditions \eqref{c1} and \eqref{c2} hold true
		if and only if
		\begin{equation}\label{remark}
			|B_{2R}|\frac{N_R}{|\Lambda|}\leq e^{-a}\min\left\{
			(\alpha-C e^{b+c})\frac{|B_{R+r}(0)\setminus B_{R-r}(0)| N_r}{|\Lambda|-N_R|B_{R+r}|},
			\frac{b |B_{2R}|}{|B_{R+r}(0)\setminus B_{R-r}(0)|}
			\right\}.
		\end{equation}
		Hence, as $R\to\infty$ for fixed $r>0$ and fixed parameters $\alpha, b$ and $c$,
		we obtain
		\begin{equation}\label{remark2}
			|B_{2R}|\frac{N_R}{|\Lambda|}\leq 
			C_1 e^{- C_2 |B_{R+r}(0)\setminus B_{R-r}(0)|\frac{N_r}{|\Lambda|-N_R |B_{R+r}|}},
		\end{equation}
		for some $C_1, C_2>0$,
		with the exponential decay being of the order of the surface of the large sphere
		rather than the volume as the standard cluster expansion condition would have implied.
	\end{remark}
	
	\begin{remark}\label{second_remark}
		A second observation arising in the previous remark is that we can also choose as
		the macroscopic density of the large spheres the quantity
		$|B_{R+r}|\frac{N_R}{|\Lambda|}$, rather than simply 
		$\frac{N_R}{|\Lambda|}$.
		This is also compatible with the validity of the sufficient conditions as discussed
		also in the previous proof. We expect that a more elaborate analysis for the
		regime $R/r\to\infty$ can be made, but this goes beyond the goal of the present article.
	\end{remark}
	
	\begin{remark}\label{third_remark}
		As discussed in Remark~\ref{key_remark} we have two main convergence regimes:
		the first is when we treat the partition function as a two-species model which leads
		to two-species, two-connected graphs. The second is
		when we first integrate over the small spheres and treat the effective system of
		the big spheres which has a more complicated multi-body structure as given in
		\eqref{ClusterEff}.
		However, if one expands further this more complicated structure, should get back
		the two-connected graphs.
		The key point here is that expanding further is equivalent to saying that we look at
		absolutely convergent series of more terms (by moving the absolute value to a larger number of terms) and this is where the improved convergence condition might be
		lost. In the next section, we look for simpler forms of the derived expansions, always making sure
		that we do not lose the improved radius of convergence.
	\end{remark}

	\section{Some cancellations}\label{s:canc}

	In the sequel we provide simpler formulas for the finite volume free energy derived
	in Proposition~\ref{ClusterEffective2}. In particular we intend to express it in terms of 
	the densities of the large and of the small spheres as well as investigate its behaviour in the limit
	of infinite volume. 
	
	\subsection{The ``adjustment'' term}\label{s_one}
	Following the analysis in Section 6 of \cite{pulvirenti2012cluster} we can simplify \eqref{ALambda} and then
	pass to the thermodynamic limit.
	By power counting of the orders of the volume, the dominant terms in \eqref{ALambda} at the thermodynamic limit 
	are those that satisfy:
	\begin{equation}
		V_i\neq V_j,\,\forall i\neq j\,\, \text{and}\label{p1.3}
	\end{equation}
	\begin{equation}
		\left|\cup_{i=1}^n V_i\right|= \sum_{i=1}^n(|V_i|-1)+1\label{p1.4}.
	\end{equation}
	Let
	\begin{equation}\begin{split}\label{Ainf}
			A_{\infty}(k;\rho):=&
			\frac{1}{k!}\sum_{g\in\mathcal C_{k+1}}\int\limits_{(\mathbb R^d)^k}\prod_{(i,j)\in E[g]} f^{ss}(q_i,q_j)\d \mathbf q_{[k]}\times\\
			&\;\;\;\;\;\sum_{n=1}^{k}\sum^*_{\substack{(V_1,\ldots,V_n):\\ \cup_{i=1}^n V_i=[k+1]}}
			\phi^T(V_1,\ldots,V_n)
			\sum_{l=1}^{k+n}\binom{k+n}{l}\left(
			\frac{\rho}{1-\rho}
			\right)^l.
		\end{split}
	\end{equation}
	with $[k]:=\{1,\ldots,k\}$, $q_{k+1}\equiv 0$ and where the $^*$ in the sum indicates that we consider
	only the terms that satisfy \eqref{p1.3} and \eqref{p1.4}.
	Then we have:
	\begin{lemma}
		For any $\rho\geq 0$ and any $k\in\mathbb N$ we have that
		\begin{equation}\label{error_A}
			|A_{\Lambda}(k;\rho)-A_{\infty}(k;\rho)|\leq C_k\frac{1}{|\Lambda|},
		\end{equation}
		where $A_{\Lambda}$ is given in \eqref{ALambda}.
	\end{lemma}
	
	\begin{proof}
		As in \eqref{est} this is a direct consequence of the analysis in Section 6 of \cite{pulvirenti2012cluster}. 
		The main idea is the following: 
		first, by a power counting of the order of the volumes $|\Lambda|$ we see that order one
		contributions are those that satisfy \eqref{p1.3} and \eqref{p1.4}.
		Second, focusing on only these terms (with an error of order $\frac{1}{|\Lambda|}$),
		given $g\in\mathcal C_{n+1}$ we can decompose it into its two-connected components $b_1,\ldots,b_m$. 
		Different choices of $V$'s can result to the same underlying graph $g$, contributing with different
		signs depending on the number of elements $V$ and different combinatorial coefficients.
		Overall, in \cite{pulvirenti2012cluster} it is proved that everything cancels and we are only left
		with one two-connected component.
		However, this is not so in our case, since on top of the combinatorial coefficients we also have the last (extra) factor in \eqref{Ainf}.
	\end{proof}
	
	\subsection{Interaction with one big}\label{s_two}
	We treat the second term of \eqref{ClusterEff}, i.e., the interaction between one big sphere and many small (in the presence of other big).
	The following structure will be relevant:
	\begin{definition}\label{Def2}
		A vertex is called {\textit{articulation vertex}} if upon its removal the component of which it is part separates into two or more connected pieces in such a way that at least one piece contains no white vertices. We denote by $\mathcal{B}^{AF}_{n,n+k}$ the subset of $\mathcal{G}_{n,n+k}$ free of articulation vertices.	
	\end{definition}
	It is worthwhile noticing the difference between cut-point and articulation free: in the case
	of no presence of white vertices they coincide.
	
	We first observe that under periodic boundary conditions we have that:
	\begin{equation}
		\int_{\Lambda}\exp\bigg\{\int_{\mathbb Y_{\Lambda}}\zeta(p,Y)\nu_{\rho_r,\Lambda}(dY)\bigg\}\frac{\d p}{|\Lambda|}\;=\;\exp\bigg\{\int_{\mathbb Y_{\Lambda}}\zeta(0,Y)\nu_{\rho_r,\Lambda}(dY)\bigg\}.
	\end{equation}
	Then, we can write this term in more detail:
	
	\begin{lemma}\label{second_cancellation1}
		Under assumption \eqref{cond1} we have:
		\begin{equation}\label{canc_one}
			\left|\int_{\mathbb{Y}_{\Lambda}}\zeta(0,Y)\nu_{\Lambda,N_r,\mathbf p}(dY)
			- \frac{N_r}{|\Lambda|}\sum_{s\geq 1}\frac{1}{s+1}P_{|\Lambda|,N_r}(s)
			B^{*,(1)}_{\Lambda}(s)\right|\leq C\frac{1}{|\Lambda|},
		\end{equation}
		where
		\begin{equation}\begin{split}\label{B1}
				B^{*,(1)}_{\Lambda}(s)  = &
				\left(\frac{|\Lambda|}{|\tilde\Lambda(\mathbf p)|}\right)^{s+1}\Big(
				B^{\emptyset}_{\Lambda}(r)\int_{\Lambda}f^{ls}(0,q)\d q +\\
				&+\frac{1}{s!}
				\sum_{b\in\mathcal B_{1,s+1}}\int_{\Lambda^{s+1}}\prod_{e\in E(b)}(\mathbf 1_{e\cap \{0\}=\emptyset}f^{ss}_{e}+\mathbf 1_{e\cap \{0\}\neq \emptyset}f^{ls}_{e})
				\prod_{i=1}^{s+1}\d q_i\Big).\end{split}
		\end{equation}
		We denote by $\mathcal B_{1,s+1}$ the two-species two-connected graphs with
		one big and $s+1$ small.
		In the thermodynamic limit $\Lambda\to\mathbb R^d$, $N_r=\lfloor \rho_r |\Lambda|\rfloor$
		and $N_R=\lfloor \rho_R |\Lambda|\rfloor$ we obtain the following upper bound:
		\begin{equation}\label{TLONE}
			\sum_{s\geq 1}\frac{1}{s+1}\rho_r^{s+1}B^{(1)}_{\infty}(s),
		\end{equation}
		where
		\begin{equation}\begin{split}\label{B1inf}
				B^{(1)}_{\infty}(s) = &
				\frac{1}{(1-\rho_R |B_{R+r}|)^{s+1}}
				\Big(
				\beta_s |B_{R+r}|+\\
				&
				+\frac{1}{s!}
				\sum_{b\in\mathcal B_{1,s+1}}\int_{(\mathbb R^d)^{s+1}}\prod_{e\in E(b)}(\mathbf 1_{e\cap \{0\}=\emptyset}f^{ss}_{e}+\mathbf 1_{e\cap \{0\}\neq \emptyset}f^{ls}_{e})
				\prod_{i=1}^{s+1}\d q_i\Big).
			\end{split}
		\end{equation}
		Similarly for the lower bound.
	\end{lemma}
	
	\begin{proof}
		By developing the products in $\zeta(p,Y)$ we obtain:
		\begin{equation}\begin{split}\label{inter0}
				\int_{\mathbb{Y}_{\Lambda}}\zeta(0,Y)\nu_{\Lambda,N_r,\mathbf p}(dY)
				& = 
				\sum_{n=1}^{\infty}\frac{1}{n!}\sum_{(V_1,\ldots,V_n)\in\mathcal{V}_{r}^n}\varphi^T(V_1,...,V_n)\left(\frac{|\Lambda|}{|\tilde\Lambda(\mathbf p)|}\right)^{\sum_{i=1}^n |V_i|} \times\\
				&\hspace{-.7cm}\int\limits_{\Lambda^{\card{V_1}}}\cdots \int\limits_{\Lambda^{\card{V_n}}}\sum_{\substack{A_1,\ldots, A_n \\ A_i\subset V_i, \cup_{i=1}^n A_i\neq\emptyset}}\prod_{i=1}^n\prod_{k\in A_i} f^{ls}(0,q_k^i)\prod_{i=1}^n \d\xi\big((q_k^i)_{k\in V_i}\big),
			\end{split}
		\end{equation}
		where we denote by $A_i\subset V_i$ the labels of the particles that interact with the big one within $V_i$.
		As it is common in the literature we call them ``white'' vertices/labels. 
		Note that there should be at least one over all $i=1,\ldots,n$.
		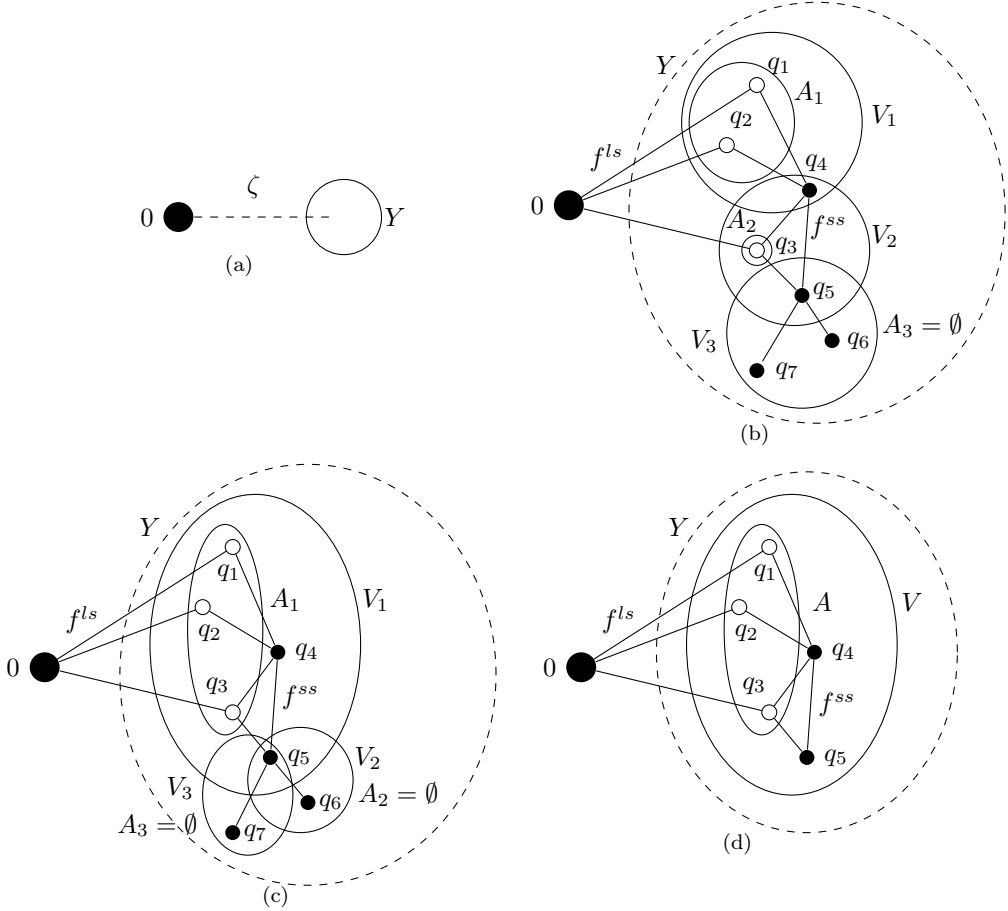
\begin{figure}[h]
			\begin{center}
				\begin{tabular}{rcl}
					&&\subfigure[]
					{\label{Fig:2.1}
						\;\;\;\;\;\;\;\;\begin{tikzpicture}
							\draw [dashed](0,0)node[bigvertex,label=left:{$0$}]{}--(2,0);
							\draw(2.2,0)circle(0.5){}node[label=right:{\;\;\;$Y$}]{};
							\draw(1,0)node[label=above:{$\zeta$}]{};
						\end{tikzpicture}
					} \;\;\;\;\;\;\;\;\;\;\;\subfigure[]
					{\label{Fig:2.2}
						\tikz[baseline=-2pt]{
							\draw(0,0)node[bigvertex,label=left:{$0$}](v0){}--(2.5,1.6)node[nuEP](V1){};
							\draw(2.8,1.5)node[label={$q_1$}]{};
							\draw(v0)--(2.1,0.8)node[nuEP](V2){};
							\draw (2.3,1.5)node[label=below:{$q_2$}]{};
							\draw(v0)--(2.5,-0.6)node[nuEP](V3){};
							\draw(2.9,-0.9)node[label={$q_3$}]{};
							\draw(0.5,0.2)node[label=above:{$f^{ls}$}]{};
							\draw(3.2,0.2)node[vertex](V4){};
							\draw(3.3,0.2)node[label={$q_4$}]{};
							\draw(3.1,-1.2)node[vertex](V5){};
							\draw(3.4,-1.5)node[label={$q_5$}]{};
							\draw(3.5,-1.8)node[vertex](V6){};
							\draw(3.45,-1.8)node[label=right:{$q_6$}]{};
							\draw(2.5,-2.2)node[vertex](V7){};
							\draw(2.5,-2.2)node[label=right:{$q_7$}](V7){};
							\draw(2.3,1.1)ellipse[x radius=0.7cm,y radius=0.8cm];
							\draw(2.7,1.1)circle(1.2);
							\draw(3.2,1.1)node[label={$A_1$}]{};
							\draw(4.2,0.8)node[label={$V_1$}]{};
							\draw(2.5,-0.6)circle(0.2);
							\draw(3,-0.6)circle(1);
							\draw(4.2,-0.8)node[label={$V_2$}]{};
							\draw(2.3,-0.6)node[label={$A_2$}]{};
							\draw(3.1,-1.7)circle(1);
							\draw(1.8,-2.2)node[label={$V_3$}]{};
							\draw(4.7,-2)node[label={$A_3=\emptyset$}]{};
							\draw(3.4,-0.7)node[label={$\;f^{ss}$}]{};
							\draw(V1)--(V4);
							\draw(V2)--(V4);
							\draw(V3)--(V4);
							\draw(V3)--(V5);
							\draw(V4)--(V5);
							\draw(V5)--(V6);
							\draw(V5)--(V7);
							\draw[dashed](3.3,-0.1)ellipse[x radius=2.5cm,y radius=2.8cm];
							\draw(1.3,1.5)node[label={$Y$}]{};
					}}
					\\
					&&\subfigure[]
					{\label{Fig:2.3}
						\hspace{-1cm}\tikz[baseline=-2pt]{
							\draw(0,0)node[bigvertex,label=left:{$0$}](v0){}--(2.5,1.6)node[nuEP](V1){};
							\draw(2.45,0.9)node[label={$q_1$}]{};
							\draw(v0)--(2.1,0.8)node[nuEP,label=below:{\;\;$q_2$}](V2){};
							\draw(v0)--(2.5,-0.6)node[nuEP](V3){};
							\draw(2.3,-0.6)node[label={$q_3$}]{};
							\draw(0.5,0.2)node[label=above:{$f^{ls}$}]{};
							\draw(2.4,0.5)ellipse[x radius=0.5cm,y radius=1.4cm];
							\draw(3.1,0.2)node[vertex,label=right:{$q_4$}](V4){};
							\draw(3,-1.2)node[vertex,label=right:{$q_5$}](V5){};
							\draw(3.5,-1.8)node[vertex](V6){};
							\draw(3.4,-1.8)node[label=right:{$q_6$}]{};
							\draw(2.5,-2.2)node[vertex](V7){};
							\draw(2.4,-2.2)node[label=right:{$q_7$}]{};
							\draw(3.2,0.5)node[label={$A_1$}]{};
							\draw(2.8,0.3)ellipse[x radius=1.4cm,y radius=2cm];
							\draw(4.4,0.5)node[label={$V_1$}]{};
							\draw(3.4,-0.8)node[label={$f^{ss}$}]{};
							\draw(3.4,-1.5)ellipse[x radius=0.7cm,y radius=0.7cm];
							\draw(4.3,-1.6)node[label={$V_2$}]{};
							\draw(4.7,-2.1)node[label={$A_2=\emptyset$}]{};
							\draw(2.7,-1.7)ellipse[x radius=0.6cm,y radius=0.8cm];
							\draw(1.8,-2)node[label={$V_3$}]{};
							\draw(1.5,-2.5)node[label={$A_3=\emptyset$}]{};
							\draw(V1)--(V4);
							\draw(V2)--(V4);
							\draw(V3)--(V4);
							\draw(V3)--(V5);
							\draw(V4)--(V5);
							\draw(V5)--(V6);
							\draw(V5)--(V7);
							\draw[dashed](3.5,-0.1)ellipse[x radius=2.5cm,y radius=2.8cm];
							\draw(1.4,1.5)node[label={$Y$}]{};
					}}\;\;\;\;\subfigure[]
					{\label{Fig:2.4}
						\tikz[baseline=-2pt]{
							\draw(0,0)node[bigvertex,label=left:{$0$}](v0){}--(2.5,1.6)node[nuEP](V1){};
							\draw(2.45,0.9)node[label={$q_1$}]{};
							\draw(v0)--(2.1,0.8)node[nuEP,label=below:{\;\;$q_2$}](V2){};
							\draw(v0)--(2.5,-0.6)node[nuEP](V3){};
							\draw(2.3,-0.6)node[label={$q_3$}]{};
							\draw(0.5,0.2)node[label=above:{$f^{ls}$}]{};
							\draw(2.4,0.5)ellipse[x radius=0.5cm,y radius=1.4cm];
							\draw(3.1,0.2)node[vertex,label=right:{$q_4$}](V4){};
							\draw(3,-1.2)node[vertex,label=right:{$q_5$}](V5){};
							\draw(3.2,0.5)node[label={$A$}]{};
							\draw(2.8,0.3)ellipse[x radius=1.4cm,y radius=2cm];
							\draw(4.4,0.5)node[label={$V$}]{};
							\draw(V1)--(V4);
							\draw(V2)--(V4);
							\draw(V3)--(V4);
							\draw(V3)--(V5);
							\draw(V4)--(V5);
							\draw(3.4,-0.9)node[label={$f^{ss}$}]{};
							\draw[dashed](3,0.2)ellipse[x radius=2cm,y radius=2.4cm];
							\draw(1.3,1.5)node[label={$Y$}]{};
					}}
				\end{tabular}
			\end{center}\caption{The cancellations for one big sphere interacting with a cloud.\newline
				\ref{Fig:2.1} The left hand side of \eqref{inter0}. \newline
				\ref{Fig:2.2} An element of \eqref{inter1} after expliciting $\zeta(0,Y)$ in terms of $f^{ls}$ links. \newline
				\ref{Fig:2.3}  Thermodynamically dominant terms: all white vertices should be in the same polymer $V$.\newline
				\ref{Fig:2.4} Articulation free graphs obtained after the cancellations.}
			\label{fig:2}
		\end{figure}	
		
		By putting together the sums over $V$ and $A$ we obtain:
		\begin{equation}\begin{split}\label{inter1}
				&\sum_{n=1}^{\infty}\frac{1}{n!}\sum_{\substack{((V_1, A_1),\ldots,(V_n,A_n))\\ \cup_{i=1}^n A_i\neq\emptyset}}\varphi^T(V_1,...,V_n)
				\left(\frac{|\Lambda|}{|\tilde\Lambda(\mathbf p)|}\right)^{\sum_{i=1}^n |V_i|}\times\\
				&\hspace{1.9cm}\int\limits_{\Lambda^{\card{V_1}}}\cdots \int\limits_{\Lambda^{\card{V_n}}}\prod_{i=1}^n\prod_{k\in A_i} f^{ls}(0,q_k^i) \prod_{i=1}^n \d\xi\big((q_k^i)_{k\in V_i}\big).
			\end{split}
		\end{equation}
		We observe that this is in the form of the abstract polymer model used in \cite{kuna2018convergence}, formula (3.4),
		with the difference that instead of $\prod_{k\in A}\phi(q_k)$ we have $\prod_{k\in A} f^{ls}(0,q_k)$
		in the activity function. 
		We recall that in Remark~\ref{r:APM} we commented on this similarity between the two partition functions.
		However, in this paper we had to first
		organize the contributions
		in terms of the interactions with the big spheres as in \eqref{key_formula}.	
		Nevertheless, once this is settled, then we can still ``discover'' the same structure as in
		\cite{kuna2018convergence}, formula (3.4), but for a different choice of test function
		as described above.
		In this lemma we treat the interactions 
		with only one big, while the case with several big
		will be investigated later in Lemma~\ref{second_cancellation2}.
		By choosing the labels of the small spheres we have:
		\begin{eqnarray}\label{B}
			\int_{\mathbb{Y}_{\Lambda}}\zeta(0,Y)\nu_{\Lambda,N_r,\mathbf p}(dY)
			&=&
			\sum_{l\geq 1}\sum_{m=1}^l\sum_{k=0}^{N_r-m}
			P_{|\Lambda|,N_r}(m+k)B^{(1)}_{\Lambda}(l,m,k),
		\end{eqnarray}
		where $m$ is the number of white vertices and $l$ the number of
		white vertices counted with their multiplicity, that is the number of times a particular vertex
		appears in different polymers $V$.
		The explicit formula of $B^{(1)}_{\Lambda}(l,m,k)$ is not relevant here, but it can be understood by comparing to \eqref{inter1}.
		Note also that for $l=1$, then necessarily $k\geq 1$.
		Furthermore, from formula (4.3) in \cite{kuna2018convergence} we have:
		\begin{equation}
			B^{(1)}_{\Lambda}(l,m,k)=\bar B^{(1)}_{\Lambda}(l,k)\delta_{l,m}+R_{\Lambda}(l,m,k),
		\end{equation}
		where, using Lemma~4.1 from \cite{kuna2018convergence}, we have:
		\begin{equation}\label{further21}
				\bar B^{(1)}_{\Lambda}(l,k) = \frac{|\Lambda|^{l+k}}{l!k!}\int_{\Lambda^{l+k}}
				\prod_{r=1}^l
				f^{ls}(0,q_r)
				\left(\frac{|\Lambda|}{|\tilde\Lambda(\mathbf p)|}\right)^{k+l}
				\sum_{g\in\mathcal B^{AF}_{l,k+l}}
				\prod_{\{i,j\}\in E(g)}f^{ss}(q_i,q_j)
				\prod_{i=1}^{l+k}\frac{\d q_i}{|\Lambda|}.
		\end{equation}
		and
		\begin{equation}\label{e1}
			|R_{\Lambda}(l,m,k)|\leq C\frac{1}{|\Lambda|},
		\end{equation}
		uniformly on its parameters.
		We do not repeat here the proof, but for the reader's convenience we present in Figure~\ref{fig:2} the main idea of how to obtain \eqref{further21}.
		In Figure~\ref{Fig:2.2} we have an element of the right hand side of \eqref{inter1}.
		As explained in detail in Section~4 in \cite{kuna2018convergence}, by power counting of the volumes $|\Lambda|$ in the various contributions we
		recognize that the leading order term $\bar B^{(1)}_{\Lambda}(l,k)\delta_{l,m}$
		should contain only those polymers in which all white vertices are in the same polymer
		as depicted in Figure~\ref{Fig:2.3}. In particular, the white vertices can not be cut-points
		(because otherwise they would have been part of two different polymers). 
		We decompose the overall graph into its two-connected components (or articulation-free
		for the component that contains the white vertices).
		By applying Lemma 4.1 in \cite{kuna2018convergence}
		the extra two-connected components with only black vertices 
		cancel obtaining the articulation-free
		graphs as in Figure~\ref{Fig:2.4}.
		
		In particular, for the case $l=1$ we observe that $\mathcal B^{AF}_{1,k+1}=\mathcal B_{k+1}$,
		hence we obtain:
		\begin{equation}\begin{split}\label{cls}
				\bar B^{(1)}_{\Lambda}(1,k)
				&= 
				\left(\frac{|\Lambda|}{|\tilde\Lambda(\mathbf p)|}\right)^{k+1}
				\frac{1}{k!}\times\\
				&\hspace{.5cm}\int_{\Lambda} f^{ls}(0,q_1)
				\left[\frac{1}{k+1}
				\int_{\Lambda^k}\sum_{g\in \mathcal B_{k+1}} \prod_{\{i,j\}\in E(g)} f^{ss}(q_i,q_j)\prod_{i=2}^{k+1} \d q_i
				\right]
				\d q_1\\
				& =  
				\left(\frac{|\Lambda|}{|\tilde\Lambda(\mathbf p)|}\right)^{k+1}
				\beta_{\Lambda}(k)
				\int_{\Lambda}f^{ls}(0,q)\d q.
			\end{split}
		\end{equation}
		On the other hand, when $l\geq 2$ and $k\geq 0$, we obtain:
		\begin{equation}\begin{split}\label{inter4}
				\bar B^{(1)}_{\Lambda}(l,k)
				&= 
				\left(\frac{|\Lambda|}{|\tilde\Lambda(\mathbf p)|}\right)^{l+k}\frac{1}{l! k!}\sum_{g\in \mathcal B^{AF}_{l,k+l}}\int_{\Lambda^{l+k}}\prod_{i=1}^l f^{ls}(0,q_i) \prod_{\{i,j\}\in E(g)} f^{ss}(q_i,q_j)\prod_{i=1}^{l+k}\d q_i\\
				&=
				\left(\frac{|\Lambda|}{|\tilde\Lambda(\mathbf p)|}\right)^{l+k}
				\frac{1}{(l+k)!}\sum_{b\in\mathcal B_{1,l+k}}\int_{\Lambda^{l+k}}\prod_{e\in E(b)}(\mathbf 1_{e\cap \{0\}=\emptyset}f^{ss}_{e}+\mathbf 1_{e\cap \{0\}\neq \emptyset}f^{ls}_{e})
				\prod_{i=1}^{l+k}\d q_i,
			\end{split}
		\end{equation}
		where we changed the combinatorial factor while going from $\mathcal B^{AF}_{l,k+l}$ to $\mathcal B_{1,l+k}$.
		Putting together these two contributions and re-arranging terms we define 
		\begin{equation}\label{Bstar1}
			B^{*,(1)}_{\Lambda}(s):=\bar B^{(1)}_{\Lambda}(1,s)+\sum_{l\geq 2}\bar B^{(1)}_\Lambda(l,s-l)
		\end{equation}
		and obtain \eqref{B1}.
		For the thermodynamic limit we recall that the cluster expansion estimate is uniform in $\Lambda$, $N_r$ and $N_R$ and by the dominated convergence theorem we can pass to the limit.
	\end{proof}
	
	\begin{remark}\label{r:lim}
		The discrepancy between the two bounds is negligible
		in the further scaling $R>>r$, which is also compatible with the analysis of Remark~\ref{key_remark}. However, a complete analysis requires more work and we hope
		to address it elsewhere.
	\end{remark}

	\begin{remark}\label{r:canc}
		Note that for consistency one can compare the terms here with
		the ones of the two-species two-connected graphs.
		In fact, in such a comparison, the first contribution for one big sphere should cancel with some part (the first order in $\rho_R |B_{R+r}|$) of the first term in
		\eqref{ClusterEff}.
		We expect that this fact can be generalized to any value $N_R$ of big spheres,
		however already for $N_R=2$ it can be seen that it gets significantly more complex.
		However, by performing the cancellation, one essentially moves the position of the absolute value 
		from outside $\zeta(0,Y)$ into the various contributions, namely distinguishing
		when $0$ has one or more $f^{ls}$-links with vertices in $Y$.
		This results to a different radius of convergence since the present proof of the cluster expansion is valid only when $\zeta(0,Y)$ is considered as a block.
	\end{remark}

	\subsection{Interaction with many big}\label{s_three}

	Following \cite{pulvirenti2012cluster}, we investigate how
	the last term in \eqref{ClusterEff}
	can be re-written as a sum over a special class of graphs.
	
	First, like before in Section~\ref{s_one}, by a power counting of the volumes (but in the measures corresponding to the big spheres this time) one can easily see that the non-negligible terms in the thermodynamic limit $\Lambda\to\mathbb R^d$ should satisfy again \eqref{p1.3} and \eqref{p1.4}.
	Furthermore, using the bipartite leaf-constrained graph representation of the hypergraphs
	we introduce the class $\mathcal B^*_{n,k}$ with $n$ large spheres and $k$ clouds
	such that no large sphere is a cut-point. For notational purposes we use the asterisk $^*$ in order
	to distinguish to the general two-species two-connected graphs as in \eqref{B1}.
	Given such a graph $g\in\mathcal B^*_{n,k}$ and a cloud indexed by $l\in \{n+1,\ldots, n+k\}$, 
	we denote by $A_g(l)$ the adjacent vertices to this cloud and by
	\begin{equation}\label{newset}
		\mathcal H(g):=\{A_g(l)\subset \{1,\ldots, n\},\forall l\in \{n+1,\ldots, n+k\}\},
	\end{equation}
	the set of indices of the adjacent vertices (representing the large spheres) to the clouds, i.e., we recover the hypergraph from its bipartite leaf-constrained representation.
	Based on this, the following lemma holds.
	
	\begin{lemma}\label{LemmaUltimo} Under the assumptions of Proposition~\ref{ClusterEffective2} we have 
		\begin{equation}\label{EffectiveFreeEFV}
				\left|\frac{1}{|\Lambda|}\sum_{n\ge1}\frac{1}{n!}\sum_{(V_1,...,V_n)\in\mathcal{V}^n_R}\varphi^T(V_1,...,V_n)\prod_{i=1}^n w_{\Lambda}(V_i) -
				\frac{N_R}{|\Lambda|}\sum_{k\ge1}\frac{1}{k+1}P_{|\Lambda|,N_R}(k)B^*_{\Lambda}(k; N_r)\right|				
				\leq C\frac{1}{|\Lambda|},
		\end{equation}		
		where $P_{|\Lambda|,N_R}(n)$ is defined in \eqref{P} and
		\begin{equation}\begin{split}
				\label{withHG}
				B^*_{\Lambda}(k;N_r)  := & \frac{|\Lambda|^k}{k!}\sum_{m\geq 0}\frac{1}{m!}\sum_{g\in\mathcal{B}^*_{k+1,m}}
				\int_{\Lambda^{k+1}}
				\prod_{\substack{\{i,j\}\in E(g)\\i,j\in\{1,...,k+1\}}}f^{ll}(p_i,p_j)
				\times\\
				&\hspace{.8cm}\prod_{J\in \mathcal H(g)}C_{\Lambda}(\mathbf p_J;N_r)
				\prod_{j=k+2}^{k+m+1}\nu_{\Lambda,N_r,\mathbf p}(dY_j)
				\prod_{i=1}^{k+1}\frac{\d p_i}{|\Lambda|},
			\end{split}
		\end{equation}
		with
		\begin{equation}\begin{split}\label{C}
				C_{\Lambda}((p_j)_{j\in J};N_r) := & \int\limits_{\mathbb Y_{\Lambda}}\prod_{j\in J}\zeta(p_j,Y)\nu_{\Lambda,N_r,\mathbf p} (dY)\\
				= &
				\frac{N_r}{|\Lambda|}\sum_{k\ge1}P_{|\Lambda|,N_r}(k)B^{\mathbf{p}}_{\Lambda}(k,J)
			\end{split}
		\end{equation}
		and $B^{\mathbf{p}}_{\Lambda}(k,J)$ given in \eqref{Bp}.
	\end{lemma}	
	
	\begin{proof}
		This is another application of the cancellations occurring in Section 5 of \cite{pulvirenti2012cluster}. We have:
		\begin{equation}\begin{split}
				&\frac{1}{|\Lambda|}\sum_{n\ge1}\frac{1}{n!}\sum_{(V_1,...,V_n)\in\mathcal{V}^n_R}\varphi^T(V_1,...,V_n)\prod_{i=1}^n w_{\Lambda}(V_i)\\
				&\hspace{.3cm}
				=\frac{N_R}{|\Lambda|}\sum_{k\ge1}\frac{1}{k+1}P_{|\Lambda|,N_R}(k)
				\frac{|\Lambda|^k}{k!}
				\sum_{n\ge1}\frac{1}{n!}\sum_{\substack{(V_1,...,V_n)\in\mathcal{V}^n_{R}\\ \bigcup_{i=1}^{n}V_i=[k+1]}}\varphi^T(V_1,...,V_n)
				\prod_{i=1}^n w_{\Lambda}(V_i).
		\end{split}\end{equation}
		Again, by a power counting of the orders of the volumes $|\Lambda|$ the
		leading order terms should satisfy \eqref{p1.3} and \eqref{p1.4}. All other terms
		are of order $\frac{1}{|\Lambda|}$ or lower.
		Then,
		by looking in the structure inside $w_{\Lambda}$ as given in \eqref{w} and \eqref{phistar} (see also Figure~\ref{fig:1}), we organize the sum over
		graphs $g\in\mathcal C^*_{k+1,m}$ for $m,k\in\mathbb N$. We observe that several choices of $V_1,\ldots,V_n$ and $n$ with $k+1=|\cup_{i=1}^n V_i|$ and \eqref{p1.3}-\eqref{p1.4}, can produce the same weight of $g$. The combinatorial coefficients make all these contributions cancel with
		each other with the remaining ones being the
		two-connected components. Thus, we obtain:
		\begin{equation}\begin{split}\label{firstB*}
				&\frac{|\Lambda|^k}{k!}
				\sum_{n\ge1}\frac{1}{n!}\sum_{\substack{(V_1,...,V_n)\in\mathcal{V}^n_{R}\\ \bigcup_{i=1}^{n}V_i=[k+1]}}\varphi^T(V_1,...,V_n)
				\prod_{i=1}^n w_{\Lambda}(V_i)=\frac{|\Lambda|^k}{k!}\times\\
				&\hspace{.7cm}\sum_{m\geq 0}\frac{1}{m!}\sum_{g\in\mathcal{B}^*_{k+1,m}}
				\int_{\Lambda^{k+1}}
				\prod_{\substack{\{i,j\}\in E(g)\\i,j\in\{1,...,k+1\}}}f^{ll}(p_i,p_j)
				\Bigg[
				\int_{\mathbb{Y}^m}\prod_{\substack{\{i,j\}\in E(g)\\i\in\{1,...,k+1\}\\j\in\{k+2,...,k+m+1\}}}\zeta(p_i,Y_j)\times\\ \\
				&\hspace{6cm}
				\times\prod_{j=k+2}^{k+m+1}\nu_{\Lambda,N_r,\mathbf p}(dY_j)\Bigg]
				\prod_{i=1}^{k+1}\frac{\d p_i}{|\Lambda|}
				+O\left(\frac{1}{|\Lambda|}\right).
		\end{split}\end{equation}	
		Then, we re-organize the indices using \eqref{newset} which gives \eqref{withHG}.
	\end{proof}
	
	Now we look at possible cancellations within the clouds interacting with a given number of big spheres. This is similar to Lemma~\ref{second_cancellation1}. 
	Given a configuration $p_j$, $j\in J\subset [N_R]$, of big spheres interacting with the same cloud 
	$Y\in\mathbb Y_{\Lambda}$ of small spheres, we define:
	\begin{equation}\label{Clim}
		C(\mathbf p_J;\rho_r) :=  \lim_{\substack{N_r\to\infty \\ \Lambda\to\mathbb R^d}} C_{\Lambda}((p_j)_{j\in J};\rho_r).
	\end{equation}
	We have:		
	\begin{lemma}\label{second_cancellation2}
		For $J\subset [N_R]$ and $C_{\Lambda}$ given in \eqref{C} we have:
		\begin{equation}\label{moreJ}
			\left|C_{\Lambda}((p_j)_{j\in J};N_r)-
			\sum_{k\geq 1}\frac{1}{k+1}\left(\frac{N_r}{|\Lambda|}\right)^{k+1}
			\sum_{l=1}^{k+1}
			\bar B_{\Lambda, J}(l,k+1-l)
			\right|
			\leq C\frac{1}{|\Lambda|},
		\end{equation}
		where
		\begin{equation}\begin{split}\label{further2}
				\bar B_{\Lambda, J}(l,k) = &\frac{1}{l!k!}\int_{\Lambda^{l+k}}
				\sum_{\substack{(A^{(j)})_{j\in J}, \\ A^{(j)}\subset \{1,\ldots,l\}, A^{(j)}\neq\emptyset \\ \cup_{j\in J} A^{(j)}=[l]}}
				\prod_{j\in J}\prod_{r\in A^{(j)}} f^{ls}(p_j,q_r)
				\times\\
				&\hspace{0.1cm}
				\left(\frac{|\Lambda|}{|\tilde\Lambda(\mathbf p)|}\right)^{k+l}
				\sum_{g\in\mathcal B^{AF}_{l,k+l}}
				\prod_{\{i,j\}\in E(g)}f^{ss}(q_i,q_j)
				\prod_{i=1}^{l+k} \d q_i.
			\end{split}
		\end{equation}
		In the limit $\Lambda\to\mathbb R^d$ and $N_r\to\infty$, we have:
		\begin{equation}\begin{split}\label{CTL}
				&C(\mathbf p_J;\rho_r) 
				=
				\sum_{l=1}^\infty
				\sum_{k=0}^{\infty}
				\frac{1}{k! l!}\left(\frac{\rho_r}{1-\rho_R |B_{R}|}\right)^{k+l}
				\times\\
				&\hspace{0.3cm}\int_{(\mathbb R^{d})^{k+l}}\sum_{\substack{(A^{(j)})_{j\in J}, \\ A^{(j)}\subset \{1,\ldots,l\}, A^{(j)}\neq\emptyset \\ |\cup_{j\in J} A^{(j)}|=l}}\prod_{j\in J}\prod_{r\in A^{(j)}} f^{ls}(p_j,q_r)
				\sum_{g\in\mathcal B^{AF}_{l,k+l}}
				\prod_{\{i,j\}\in E(g)}f^{ss}(q_i,q_j)
				\prod_{i=1}^{k+l} \d q_i.
			\end{split}
		\end{equation}
	\end{lemma}
	
	\begin{proof}
		Recalling \eqref{Bp} we have that
		\begin{equation}\label{C2}
			C_{\Lambda}((p_j)_{j\in J};N_r) = \frac{N_r}{|\Lambda|}\sum_{k\ge1}P_{|\Lambda|,N_r}(k)B^{\mathbf{p}}_{\Lambda}(k,J).
		\end{equation}
		By expanding $\zeta(p_j,(\mathbf q_{V_i})_{i=1}^n)$ in \eqref{Bp} we obtain:
		\begin{equation}\begin{split}\label{Bp2}
				&B^{\mathbf{p}}_{\Lambda}(k,J) =  \frac{|\Lambda|^k}{k!}
				\sum_{n\ge1}\frac{1}{n!}\sum_{\substack{(V_1,...,V_n)\in\mathcal{V}^n_{r}\\ \bigcup_{i=1}^{n}V_i=[k+1]}}\varphi^T(V_1,...,V_n)
				\left(\frac{|\Lambda|}{|\tilde\Lambda(\mathbf p)|}\right)^{\sum_{i=1}^n|V_i|}\times\\
				&\hspace{1.3cm}\int\limits_{\Lambda^{\card{V_1}}}\cdots \int\limits_{\Lambda^{\card{V_n}}}\prod_{j\in J}
				\sum_{\substack{((A^{(j)}_i)_{j\in J})_{i=1}^n, \\ A^{(j)}_i\subset V_i, \cup_i A^{(j)}_i\neq\emptyset,\forall j\in J}}\prod_{i=1}^n\prod_{k\in A^{(j)}_i} f^{ls}(p_j,q_k^i)\prod_{i=1}^n \d\xi\big((q_k^i)_{k\in V_i}\big),
		\end{split}\end{equation}
		see Figure~\ref{Fig:3.2}.
		Here, we have to repeat the computations in Lemma~\ref{second_cancellation1} for the case of several
		big spheres with centers $p_j$, $j\in J\subset [N_R]$.
		Hence, following \eqref{inter1}, 
		instead of only the pair $(V,A)$ we have now $(V,(A^{(j)})_{j\in J})$, i.e., a collection of sets 
		$A^{(j)}\subset V$ for
		the labels of the small spheres within $V$ that interact with each big sphere with center $p_j$, $j\in J$.
		The rest is similar. 
		\begin{figure}[h]
			\begin{center}
				\begin{tabular}{rcl}
					&&\subfigure[]
					{\label{Fig:3.1}
						\hspace{-1cm}\tikz[baseline=-2pt]{ 
							\draw [dashed](0,0)node[bigvertex,label=left:{$p_1$}]{}--(1.9,0);
							\draw(2.2,0)circle(0.5){};
							\draw(1.5,-0.8)node[label=right:{\;\;\;$Y$}]{};
							\draw(1,0)node[label=above:{$\zeta$}]{};
							\draw[dashed](2.5,0)--(4.2,0)node[bigvertex,label=right:{$p_3$}]{};
							\draw[dashed](2.2,2)node[bigvertex,label=right:{$p_2$}]{}--(2.2,0.3);}}
					\;\;\;\;
					\subfigure[]
					{
						\tikz[baseline=-2pt]{(\draw(-0.5,0)node[bigvertex,label=left:{$p_1$}]{}--(1.5,0.3)node[nuEP](V1){};
							\draw(1.15,0.1)node[label={$q_1$}]{};
							\draw(-0.5,0)--(1.4,-0.2)node[nuEP]{};
							\draw(1.05,-0.4)node[label={$q_2$}]{};
							\draw(0,-0.1)node[label={$f^{ls}$}]{};
							\draw(2.1,-1)node[vertex](V4){};
							\draw(1.7,-1.3)node[label={$q_6$}]{};
							\draw(2,2.3)node[bigvertex,label=right:{$p_2$}](v0){}--(1.8,0.5)node[nuEP](V2){};
							\draw(2.15,0.1)node[label={$q_3$}]{};
							\draw(4.9,0)node[bigvertex,label=right:{$p_3$}]{}--(3,0)node[nuEP](V3){};
							\draw(3.4,-0.2)node[label={$q_5$}]{};
							\draw(v0)--(3,0.5)node[nuEP]{};
							\draw(2.65,0)node[label={$q_4$}]{};
							\draw(V1)--(v0);
							\draw(V1)--(V4);
							\draw(1.5,-0.25)--(V4);
							\draw(V2)--(V4);
							\draw(3,0.4)--(V4);
							\draw(V3)--(V4);
							\draw(2,-2)node[vertex](V5){};
							\draw(2.8,-2.2)node[label=left:{$q_7$}]{};
							\draw(V4)--(V5);
							\draw(2.3,-1.6)ellipse[x radius=0.6cm, y radius=0.8cm];
							\draw(2.5,-1.5)node[vertex](V6){};
							\draw(2.55,-2.25)node[label={$q_8$}]{};
							\draw(V4)--(V6);
							\draw(3.2,-2)node[label={$V_3$}]{};
							\draw(1,-2.3)node[label={$A_3=\emptyset$}]{};
							\draw(3,0.5)circle(0.2);
							\draw(3,0)circle(0.2);
							\draw[rotate=-35](2.4,1.5)ellipse[x radius=0.7cm,y radius=1.3cm];
							\draw(3.55,0.3)node[label={$A^{(2)}_2$}]{};
							\draw(3.15,-0.9)node[label={$A^{(3)}_2$}]{};
							\draw(3.75,-0.85)node[label={$V_2$}]{};
							\draw(2.3,-0.7)node[label={$f^{ss}$}]{};
							\draw [rotate=-15](1.4,0.45)ellipse[x radius=0.2cm,y radius=0.45cm];
							\draw(1.3,-1.1)node[label={$A^{(1)}_1$}]{};
							\draw[rotate=-50](0.75,1.5)ellipse[x radius=0.2cm,y radius=0.4cm];
							\draw(1.45,0.35)node[label={$A_1^{(2)}$}]{};
							\draw[rotate=25](1.45,-0.85)ellipse[x radius=0.9cm,y radius=1.4cm];
							\draw(0.9,-1.5)node[label={$V_1$}]{};
							\draw[dashed](2.3,-0.4)circle(2.1){};
							\draw(0.8,1.2)node[label={$Y$}]{};
						}\label{Fig:3.2}}
					\\
					\\
					&&\subfigure[]{
						\hspace{-1cm}\tikz[baseline=-2pt]{(\draw(-0.5,0)node[bigvertex,label=left:{$p_1$}]{}--(1.4,0.2)node[nuEP](V1){};
							\draw(1.1,0.1)node[label={$q_1$}]{};
							\draw(-0.5,0)--(1.4,-0.2)node[nuEP]{};
							\draw(1,-0.7)node[label={$q_2$}]{};
							\draw(0,-0.1)node[label={$f^{ls}$}]{};
							\draw(2.1,-1)node[vertex](V4){};
							\draw(2.5,-1.4)node[label={$q_6$}]{};
							\draw(2,2.3)node[bigvertex,label=right:{$p_2$}](v0){}--(1.7,0.4)node[nuEP](V2){};
							\draw(1.95,-0.1)node[label={$q_3$}]{};
							\draw(4.9,0)node[bigvertex,label=right:{$p_3$}]{}--(2.9,0)node[nuEP](V3){};
							\draw(3,0)node[label={$q_5$}]{};
							\draw(v0)--(2.3,0.5)node[nuEP]{};
							\draw(2.6,0.3)node[label={$q_4$}]{};
							\draw(2.9,0)circle(0.2);
							\draw[rotate=20](1.8,-0.3)ellipse[x radius=0.7cm, y radius=0.3cm]{};
							\draw(1.4,0)ellipse[x radius=0.2, y radius=0.4]{};
							\draw(1.4,-1.2)node[label={$A_1^{(1)}$}]{};
							\draw(2.1,0.5)node[label={$A_1^{(2)}$}]{};
							\draw(3,-0.9)node[label={$A_1^{(3)}$}]{};
							\draw(V1)--(v0);
							\draw(V1)--(V4);
							\draw(1.5,-0.25)--(V4);
							\draw(V2)--(V4);
							\draw(2.3,0.4)--(V4);
							\draw(V3)--(V4);
							\draw(2,-1.5)ellipse[x radius=0.3cm, y radius=0.8cm];
							\draw(2,-2)node[vertex](V5){};
							\draw(2.8,-2.2)node[label=left:{$q_7$}]{};
							\draw(1.4,-1.9)node[label={$V_2$}]{};
							\draw(1.1,-2.4)node[label={$A_2=\emptyset$}]{};
							\draw(V4)--(V5);
							\draw(2.35,-1.25)circle(0.5);
							\draw(2.5,-1.5)node[vertex](V6){};
							\draw(2.55,-2.25)node[label={$q_8$}]{};
							\draw(V4)--(V6);
							\draw(3.1,-1.7)node[label={$V_3$}]{};
							\draw(3.45,-2.1)node[label={$A_3=\emptyset$}]{};
							\draw(2.45,-0.7)node[label={$f^{ss}$}]{};
							\draw(2.1,0)circle(1.35){};
							\draw(3.3,0.6)node[label={$V_1$}]{};
							\draw[dashed](2.3,-0.4)circle(2.1){};
							\draw(0.8,1.2)node[label={$Y$}]{};
						}\label{Fig:3.3}}\;\;\;\;
					\subfigure[]{
						\tikz[baseline=-2pt]{(\draw(-0.5,0)node[bigvertex,label=left:{$p_1$}]{}--(1.4,0.2)node[nuEP](V1){};
							\draw(1.1,0.1)node[label={$q_1$}]{};
							\draw(-0.5,0)--(1.4,-0.2)node[nuEP]{};
							\draw(1,-0.7)node[label={$q_2$}]{};
							\draw(0,-0.1)node[label={$f^{ls}$}]{};
							\draw(2.1,-1)node[vertex](V4){};
							\draw(2.4,-1.4)node[label={$q_6$}]{};
							\draw(2,2.3)node[bigvertex,label=right:{$p_2$}](v0){}--(1.7,0.4)node[nuEP](V2){};
							\draw(2,-0.05)node[label={$q_3$}]{};
							\draw(4.5,0)node[bigvertex,label=right:{$p_3$}]{}--(2.9,0)node[nuEP](V3){};
							\draw(3,0)node[label={$q_5$}]{};
							\draw(v0)--(2.3,0.5)node[nuEP]{};
							\draw(2.6,0.3)node[label={$q_4$}]{};
							\draw(2.9,0)circle(0.2);
							\draw[rotate=20](1.8,-0.3)ellipse[x radius=0.7cm, y radius=0.3cm]{};
							\draw(1.4,0)ellipse[x radius=0.2, y radius=0.4]{};
							\draw(1.4,-1.2)node[label={$A_1^{(1)}$}]{};
							\draw(2.1,0.5)node[label={$A_1^{(2)}$}]{};
							\draw(3,-0.9)node[label={$A_1^{(3)}$}]{};
							\draw(V1)--(v0);
							\draw(V1)--(V4);
							\draw(1.5,-0.25)--(V4);
							\draw(V2)--(V4);
							\draw(2.3,0.4)--(V4);
							\draw(V3)--(V4);
							\draw(2.45,-0.7)node[label={$f^{ss}$}]{};
							\draw(2.1,0)circle(1.35){};
							\draw(3.3,0.6)node[label={$V$}]{};
							\draw[dashed](2.1,0)circle(1.8){};
							\draw(3.9,-1.4)node[label={$Y$}]{};
						}\label{Fig:3.4}}
				\end{tabular}
			\end{center}\caption{The cancellations for three big spheres interacting with a cloud.\newline
				\ref{Fig:3.1} The term $C_{\Lambda}((p_1,p_2,p_3);N_r)$\newline
				\ref{Fig:3.2} One element in \eqref{Bp2}. \newline
				\ref{Fig:3.3}  Thermodynamically dominant terms: all white vertices should be in the same polymer $V$.\newline
				\ref{Fig:3.4} Articulation free graphs obtained after the cancellations.}
			\label{fig:3}
		\end{figure}
		
		We re-organize the sum over $(V_i,(A_i^{(j)})_{j\in J})_{i=1,\ldots,n}$
		and obtain:
		\begin{equation}\begin{split}\label{inter6}
				B^{\mathbf{p}}_{\Lambda}(k,J)  = &
				\frac{|\Lambda|^k}{k!}
				\sum_{n\geq1}\frac{1}{n!}\sum_{\substack{(V_i,(A_i^{(j)})_{j\in J})_{i=1,\ldots,n}\\ \cup_{i=1}^n A^{(j)}_i\neq\emptyset, \forall j\in J\\ \bigcup_{i=1}^{n}V_i=[k+1]}}\varphi^T(V_1,...,V_n)
				\left(\frac{|\Lambda|}{|\tilde\Lambda(\mathbf p)|}\right)^{\sum_{i=1}^n |V_i|}\times\\
				&\hspace{1.8cm}
				\int\limits_{\Lambda^{\card{V_1}}}\cdots \int\limits_{\Lambda^{\card{V_n}}}\prod_{j\in J}\prod_{i=1}^n\prod_{r\in A^{(j)}_i} f^{ls}(p_j,q_r^i) \prod_{i=1}^n \d\xi\big((q_k^i)_{k\in V_i}\big).
			\end{split}
		\end{equation}
		In order to apply the result from \cite{kuna2018convergence}, Section 4,
		we separate over white and black vertices. Note that it is better to do it directly for \eqref{C}, rather than $B^{\mathbf p}_{\Lambda}$. The interested reader can check how the two are related.
		Hence, following the exact steps of Lemma~\ref{second_cancellation1}
		we have:
		\begin{equation}\label{d}
			\int_{\mathbb{Y}}\prod_{j\in J}\zeta(p_j,Y)\nu_{\Lambda,N_r,\mathbf p}(dY)
			=\sum_{l\geq 1}\sum_{m=1}^l\sum_{k=0}^{N_r-m}
			P_{|\Lambda|,N_r}(m+k)B_{\Lambda,J}(l,m,k),
		\end{equation}
		where the formula of $B_{\Lambda,J}(l,m,k)$ can be understood by comparing to \eqref{inter6} and \eqref{C}.
		Here $m$ is the number of white vertices and $l$ the number of
		white vertices counted with their multiplicity, that is the number of times a particular vertex
		appears in different sets $A\subset V$.
		Furthermore, from formula (4.3) in \cite{kuna2018convergence} we have:
		\begin{equation}
			B_{\Lambda,J}(l,m,k)=\bar B_{\Lambda, J}(l,k)\delta_{l,m}+R_{\Lambda}(l,m,k),
		\end{equation}
		where in the first term we include all terms
		that contain only one polymer with $A^{(j)}\neq\emptyset$, $\forall j\in J$.
		We choose the labels $\cup_{j\in J}A^{(j)}=\{1,\ldots,l\}$ (see Figure~\ref{Fig:3.3} with $l=5$).
		Following Lemma~4.1 in \cite{kuna2018convergence} we have (Figure~\ref{Fig:3.4}):
		\begin{equation}\begin{split}\label{further20}
				\bar B_{\Lambda, J}(l,k)  = & \;\frac{|\Lambda|^{l+k}}{l!k!}\int_{\Lambda^{l+k}}
				\sum_{\substack{(A^{(j)})_{j\in J} \\ A^{(j)}\subset \{1,\ldots,l\}, A^{(j)}\neq\emptyset \\ \cup_{j\in J} A^{(j)}=[l]}}
				\prod_{j\in J}\prod_{r\in A^{(j)}} f^{ls}(p_j,q_r)
				\times\\
				&
				\left(\frac{|\Lambda|}{|\tilde\Lambda(\mathbf p)|}\right)^{k+l}
				\sum_{g\in\mathcal B^{AF}_{l,k+l}}
				\prod_{\{i,j\}\in E(g)}f^{ss}(q_i,q_j)
				\prod_{i=1}^{l+k}\frac{\d q_i}{|\Lambda|}.
			\end{split}
		\end{equation}
		Furthermore,
		\begin{equation}\label{e2}
			|R_{\Lambda}(l,m,k)|\leq C\frac{1}{|\Lambda|},
		\end{equation}
		uniformly on its parameters, which concludes
		the proof of the lemma.
	\end{proof}
	
	\begin{remark}\label{r:canc2}
		Note that as in \eqref{cls} 
		one of the cases is when all big spheres interact with the same small (which becomes white, i.e., $l=m=1$)
		obtaining:
		\begin{equation}\label{manybigonewhite}
			C_{\Lambda}((p_j)_{j\in J};N_r)=
			\sum_{k\geq 1}\frac{1}{k+1}\left(\frac{N_r}{|\tilde\Lambda(\mathbf p)|}\right)^{k+1}\beta_{\Lambda}(k)
			\int_{\Lambda}\prod_{j\in J}f^{ls}(p_j,q)\d q.
		\end{equation}
		As discussed in Remark~\ref{r:canc}, under a weaker radius of convergence, the terms in \eqref{cls} cancel at infinite volume with the order $\rho_R$ contribution of \eqref{Ainf}.
		The same should be true for the contributions of order $\rho_R^m$, with $m\geq 2$:
		they should cancel with \eqref{manybigonewhite} at infinite volume after integrating over $p_j$, $j\in J$.
		One could continue and check how at the end we reduce to two-species-two-connected graphs, under
		a weaker radius of convergence under which the required graph operations are valid.
	\end{remark}
	
	Summarizing, the previous lemmas revealed that in the final expression of
	the coefficients $B^*_{\Lambda}$ given by \eqref{withHG} and \eqref{moreJ}-\eqref{further2}
	we have an ``outside'' two-connected structure
	for the bipartite graph and an ``inside'' articulation free structure within the
	clouds consisting of white vertices for those that interact with the neighbouring big spheres
	and black ones for the rest,
	as depicted in Figure~\ref{fig:4}.
	
	\begin{figure}[h]
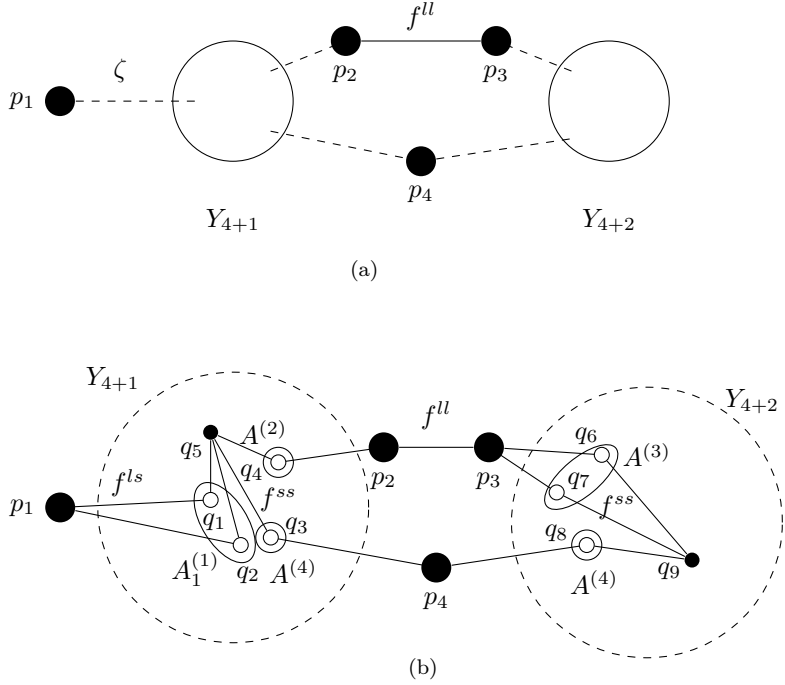

		\begin{center}
			\begin{tabular}{rcl}
				&&\subfigure[]{
					\hspace{-1cm}\tikz[baseline=-2pt]{\draw(1.5,0)circle(0.8);
						\draw(1.5,-2)node[label={$Y_{4+1}$}]{};
						\draw(6.5,0)circle(0.8);
						\draw(6.5,-2)node[label={$ Y_{4+2}$}]{};
						\draw[dashed](2,-0.4)--(4,-0.8)node[bigvertex,label=below:{$p_4$}](v0){};
						\draw[dashed](v0)--(6.,-0.5);
						\draw[dashed](2,0.4)--(3,0.8)node[bigvertex,label=below:{$p_2$}](v1){};
						\draw[dashed](6,0.4)--(5,0.8)node[bigvertex,label=below:{$p_3$}](v2){};
						\draw(v1)--(v2);
						\draw(4,0.7)node[label={$f^{ll}$}]{};
						\draw[dashed](-0.8,0)node[bigvertex,label=left:{$p_1$}](v3){}--(1,0);
						\draw(0,0)node[label={$\zeta$}]{};
					}\label{Fig:4.1}}
				\\
				\\
				\\
				&&	\subfigure[]{
					\hspace{-1cm}\tikz[baseline=-2pt]{
						\draw(6,-0.5)circle(0.2);
						\draw(6.1,-1.4)node[label={$A^{(4)}$}]{};
						\draw[rotate=-50](3.5,4.8)ellipse[x radius=0.25cm, y radius=0.6cm]{};
						\draw(6.8,0.3)node[label={$A^{(3)}$}]{};
						\draw(1.8,-0.4)circle(0.2);
						\draw(2.1,-1.25)node[label={$A^{(4)}$}]{};
						\draw(1.9,0.6)circle(0.2);
						\draw(1.7,0.6)node[label={$A^{(2)}$}]{};
						\draw[rotate=32](0.9,-0.8)ellipse[x radius=0.3cm, y radius=0.6cm]{};
						\draw(0.8,-1.3)node[label={$A_1^{(1)}$}]{};
						\draw(1.9,-0.3)node[label={$f^{ss}$}]{};
						\draw(1.8,-0.4)node[nuEP](v24){}--(4,-0.8)node[bigvertex,label=below:{$p_4$}](v0){};
						\draw(2.15,-0.6)node[label={$q_3$}]{};
						\draw(v0)--(6,-0.5)node[nuEP](v11){};
						\draw(5.6,-0.7)node[label={$q_8$}]{};
						\draw(1.9,0.6)node[nuEP](v23){}--(3.3,0.8)node[bigvertex,label=below:{$p_2$}](v1){};
						\draw(1.55,0.1)node[label={$q_4$}]{};
						\draw(5.6,0.2)node[nuEP](v12){}--(4.7,0.8)node[bigvertex,label=below:{$p_3$}](v2){};
						\draw(5.9,0)node[label={$q_7$}]{};
						\draw(v1)--(v2);
						\draw(4,0.8)node[label={$f^{ll}$}]{};
						\draw(v2)--(6.2,0.7)node[nuEP](v13){};
						\draw(6,0.6)node[label={$q_6$}]{};
						\draw(7.4,-0.7)node[vertex](V0){};
						\draw(7.1,-1.2)node[label={$q_9$}]{};
						\draw(V0)--(v11);
						\draw(V0)--(v12);
						\draw(V0)--(v13);
						\draw(6.4,-0.4)node[label={$f^{ss}$}]{};
						\draw(-1,0)node[bigvertex,label=left:{$p_1$}](v3){}--(1,0.1)node[nuEP](v21){};
						\draw(1.05,-0.55)node[label={$q_1$}]{};
						\draw(v3)--(1.4,-0.5)node[nuEP](v22){};
						\draw(1.5,-1.25)node[label={$q_2$}]{};
						\draw(-0.1,-0.1)node[label={$f^{ls}$}]{};
						\draw(1,1)node[vertex](V1){};
						\draw(0.75,0.4)node[label={$q_5$}]{};
						\draw(v21)--(V1);
						\draw(v22)--(V1);
						\draw(v23)--(V1);
						\draw(v24)--(V1);
						\draw[dashed](1.3,0)circle(1.8);
						\draw(-0.3,1.3)node[label={$Y_{4+1}$}]{};
						\draw[dashed](6.8,-0.2)circle(1.8);
						\draw(8.2,1)node[label={$Y_{4+2}$}]{};
					}\label{Fig:4.2}}
			\end{tabular}
		\end{center}\caption{\ref{Fig:4.1} A graph in $\mathcal{B}^*_{3+1,2}$ contributing to $B_{\Lambda}^*(3;N_r)$.\newline
			\ref{Fig:4.2} A contribution to the factors $C_{\Lambda}((p_1,p_2,p_4);N_r)$ and $C_{\Lambda}((p_3,p_4);N_r)$ for the graph in $\mathcal{B}^*_{3+1,2}$ given in~\ref{Fig:4.1}. }
		\label{fig:4}
	\end{figure}

	As a consequence, the thermodynamic limit of the 
	third term of \eqref{new_ub} reads:
	\begin{equation}\label{thirdTL}
		\sum_{n\geq 1}\rho_R^{n+1} B^*(n;\rho_r),
	\end{equation}
	where
	\begin{equation}\label{BTLwithHG}
			B^*(n;\rho_r)  :=  \frac{1}{n!}\sum_{k\ge0}\frac{1}{k!}\sum_{g\in\mathcal{B}^*_{n+1,k}}
			\int_{(\mathbb R^d)^{n}}
			\prod_{\substack{\{i,j\}\in E(g)\\i,j\in\{1,...,n+1\}}}f^{ll}(p_i,p_j)\prod_{J\in \mathcal H(g)}C(\mathbf p_J;\rho_r)
			\prod_{i=1}^{n}\d p_i
	\end{equation}
	and $p_{n+1}\equiv 0$.
	The factor $C(\mathbf p_J;\rho_r)$ is given in \eqref{CTL}.
	Overall, from \eqref{EffectiveFVfreeE} and the upper bound
	\eqref{new_ub}
	we obtain the following limit of the upper bound of the free energy:
	\begin{equation}\label{fe}
		\rho_r(\log\rho_r-1)+\rho_R(\log\rho_R-1)
		- F(\rho_r,\rho_R),
	\end{equation}
	where $F=F^{(0)}+A+F^{(1)}+F^{(\geq 2)}$ with
	\begin{equation}
		F^{(0)}(\rho_r)  :=  \sum_{n\geq 1}\frac{1}{n+1}\beta_n\rho_r^{n+1},\label{fe0}
	\end{equation}
	\begin{equation}
		A(\rho_r,\rho_R)  :=  \sum_{n\geq 1}\frac{1}{n+1}\rho_r^{n+1}A_{\infty}(n;\rho_R|B_{R+r}|),\label{A}
	\end{equation}
	\begin{equation}
		F^{(1)}(\rho_r,\rho_R)  :=  \rho_R\sum_{k\geq 1}\frac{1}{k+1}(\frac{\rho_r}{1-\rho_R |B_{R+r}|})^{k+1}B^{(1)}_{\infty}(k)
		\label{fe1}
	\end{equation}
and
	\begin{equation}
		F^{(\geq 2)}(\rho_r,\rho_R) := 
		\sum_{n\geq 1}\frac{1}{n+1}\rho_R^{n+1} B^*(n;\rho_r)\label{fe2}.
	\end{equation}
	Similarly, we obtain limit of the lower bound by replacing $|B_{R+r}|$ by $|B_R|$.
	Hence, as mentioned in Remark~\ref{r:lim} we 
	expect that in the further limit $R/r\to\infty$ the upper and lower bound might agree. Note also that this regime is valid within the estimates
	of the cluster expansion presented here. In such a case, obtaining an infinite volume
	canonical free energy, one can further ask
	whether we can get it by virial inversion of the corresponding effective grand canonical pressure as constructed in \cite{jansen2020cluster} in a way that maintains 
	the improved convergence condition. The aim would be to understand whether working
	directly in the canonical ensemble can provide a better expansion than working in
	the grand-canonical and subsequently performing a virial inversion which adds further constraints in the convergence condition.
	Note that this is not the case for the one-species case. For the two species with finite
	$R$ and $r$ studied here we only got upper and lower bounds bringing us to
	intriguing questions requesting further investigation.

	\appendix

	\section{Proof of the cluster expansion}
	
	Recalling the notation from Section~\ref{s_small}, given $J\subset [N_R]$ we consider the following abstract polymer model:
	\begin{equation}\label{newPF}
		Z(J):=\sum_{\{V_1,\ldots,V_n\}_{\sim}}\prod_{i=1}^n \zeta(V_i,J),
	\end{equation}
	where
	\begin{equation}\label{newact}
		\zeta(V,J):=\int_{\Lambda^V}\d\xi(\underline q_V)\prod_{j\in J}\prod_{k\in V}(1+f^{ls}(p_j,q_k)).
	\end{equation}
	We have:
	\begin{prop}\label{propCE}
		Suppose that there exists a function $a:\mathcal V_r\to\mathbb R$ such that
		\begin{equation}\label{unif_cond}
			\sup_{J\subset [N_R]}\sum_{V'\cap V\neq\emptyset} |\zeta(V',J)| e^{a(V')+c(V')}\leq a(V).
		\end{equation}
		Then
		\begin{equation}\label{rewrite}
			\ln Z(N_R)=\sum_{J\subset [N_R]} \Phi^T(J),
		\end{equation}
		where
		\begin{equation}\label{PhiT}
			\Phi^T(J)=\sum_{J'\subset J}(-1)^{|J\setminus J'|}\ln Z(J').
		\end{equation}
	\end{prop}
	
	\begin{proof}
		It follows from the Theorem in \cite{kotecky1986cluster}.
	\end{proof}

	\begin{remark}\label{r:CE}
		The factor in the big parenthesis in \eqref{att} 
		can be viewed as an abstract polymer model
		consisting of the space $\mathcal V_r$ and the activity
		$\zeta^{\mathbf p}_{\Lambda}$. Given the formulation in this appendix 
		it is also equal to $Z([N_R])$.
		Hence, we can apply Proposition~\ref{propCE} after checking
		the validity of the convergence condition \eqref{unif_cond}. This is given in the following lemma.
	\end{remark}

	For the partition function in \eqref{att} we have:

	\begin{lemma}\label{PropXuan2}
		Assume that there exist constants $b,c>0$ such that
		\begin{equation}\label{cond1}
			2\frac{N_r}{|\Lambda|-N_R |B_{R+r}|}|B_{2R}|e^{2(b+c)+1}<c.
		\end{equation}
		Then the following series is converging absolutely and uniformly in $N_r$, $\mathbf p$, $R$, $r$ and $\Lambda$:
		\begin{equation}\label{ClusterZps2}
			\sup_{V_1\subset [N_r]}
			\sum_{n\geq 2}\frac{1}{(n-1)!}
			\sum_{\substack{(V_2,\ldots,V_n)\\ V_i\in\mathcal V_r,\, |V_i|\geq 2}}
			|\phi^T(V_1,\ldots,V_n)|
			\prod_{i=2}^n|\zeta^{\mathbf p}_{\Lambda}(V_i)|\leq e^c-1
		\end{equation}	
		or, equivalently (recalling the notation in \eqref{DefYMeas}),
		\begin{equation}\label{ClusterZps3}
			\frac{1}{|\Lambda|}
			\int_{\mathbb Y_{\Lambda}}
			\d |\nu|_{\Lambda,N_r,\mathbf p}(Y)<\infty
			\quad\text{and}\quad
			\frac{1}{|\Lambda|}
			\int_{\mathbb Y_{\Lambda}}
			\sum_{\substack{J\subset [N_R] \\J\neq\emptyset}}\prod_{j\in J} 
			|\zeta(p_j,Y)|
			\d |\nu|_{\Lambda,N_r,\mathbf p}(Y)<\infty.
		\end{equation}
	\end{lemma}
	
	\begin{proof}
		Given the fact that
		$\sup_{V\subset [N_r]}\sup_{\mathbf p, \mathbf q_V}|\vartheta_{\mathbf p}(\mathbf q_V)+1|\leq 1$, we have:
		\begin{equation}\begin{split}\label{conv_cond}
				\sum_{V\ni 1}|\zeta^{\mathbf p}_{\Lambda}(V)|e^{(b+c)|V|} 
				& \leq 
				\sum_{n\geq 2}\binom{N_r-1}{n-1}\frac{1}{|\tilde\Lambda(\mathbf p)|^{n-1}}n^{n-2}|B_{2R}|^{n-1}e^{(b+c)n}\\
				& \leq 
				e^{(b+c)} \sum_{n\geq 2} \frac{n^{n-2}}{(n-1)!}\left(\frac{N_r |B_{2R}|}{|\Lambda|-N_R |B_{R+r}|}e^{b+c}\right)^{n-1}<c,
			\end{split}
		\end{equation}
		under condition \eqref{cond1}.
		Then, since $\zeta(V,[N_R])=\zeta^{\mathbf p}_{\Lambda}(V)$ the convergence
		condition \eqref{cond1} is satisfied and in view of \eqref{rewrite} and \eqref{PhiT}, 
		equation \eqref{calc} can be written as follows:
		\begin{equation}\begin{split}
				&\frac{1}{|\Lambda|}\ln Z^{\mathbf{p}}_{\Lambda,N_r} \nonumber\\
				&\hspace{.3cm}= 
				\frac{N_r}{|\Lambda|}\ln\frac{|\tilde\Lambda(\mathbf{p})|}{|\Lambda|}
				+\frac{1}{|\Lambda|}
				\sum_{n\geq 1}\frac{1}{n!}
				\sum_{\substack{(V_1,\ldots,V_n)\\ V_i\in\mathcal V_r,\, |V_i|\geq 2}}
				\phi^T(V_1,\ldots,V_n)
				\prod_{i=1}^n\zeta_{\Lambda}(V_i,[N_R])\nonumber\\
				& \hspace{.3cm}=  
				\frac{N_r}{|\Lambda|}\ln\frac{|\tilde\Lambda(\mathbf{p})|}{|\Lambda|}
				+\frac{1}{|\Lambda|}
				\sum_{n\geq 1}\frac{1}{n!}
				\sum_{\substack{(V_1,\ldots,V_n)\\ V_i\in\mathcal V_r,\, |V_i|\geq 2}}
				\phi^T(V_1,\ldots,V_n)
				\prod_{i=1}^n\zeta_{\Lambda}(V_i,\emptyset)\nonumber\\
				&\hspace{.8cm}
				+\sum_{\substack{J\subset [N_R] \\ J\neq\emptyset}}\sum_{J'\subset J}
				(-1)^{|J\setminus J'|}
				\frac{1}{|\Lambda|}
				\sum_{n\geq 1}\frac{1}{n!}
				\sum_{\substack{(V_1,\ldots,V_n)\\ V_i\in\mathcal V_r,\, |V_i|\geq 2}}
				\phi^T(V_1,\ldots,V_n)
				\prod_{i=1}^n\zeta_{\Lambda}(V_i;\,J')\nonumber\\
				&\hspace{.3cm} = 
				\frac{N_r}{|\Lambda|}\ln\frac{|\tilde\Lambda(\mathbf{p})|}{|\Lambda|}
				+\frac{1}{|\Lambda|}
				\int_{\mathbb Y_{\Lambda}}\d\nu_{\Lambda,N_r,\mathbf p}(Y)+
				\frac{1}{|\Lambda|}
				\int_{\mathbb Y_{\Lambda}}
				\sum_{\substack{J\subset [N_R] \\J\neq\emptyset}}\prod_{j\in J} \zeta(p_j,Y)
				\d\nu_{\Lambda,N_r,\mathbf p}(Y),
			\end{split}
		\end{equation}
		where in the last equality we used the fact that
		for any $J\subset [N_R]$ we have:
		\begin{equation}\label{expand}
			\prod_{j\in J}\zeta(p_j,(\mathbf q_{V_i})_{i=1}^n)  
			=
			\sum_{J'\subset J}
			(-1)^{|J\setminus J'|}\prod_{i=1}^n\prod_{j\in J'}
			(1+\bar\theta(p_j,\mathbf q_{V_i})).
		\end{equation}
		Hence, we obtained \eqref{calc2} in an alternative way.
	\end{proof}
	
	\subsubsection*{Acknowledgments}
	We acknowledge fruitful discussions with Sabine Jansen and Tobias Kuna.
	T. X. Nguyen would like to acknowledge the support of the NYU-ECNU Institute of Mathematical Sciences at NYU Shanghai.
	G. Scola acknowledges financial support from MIUR, PRIN 2017 project MaQuMA,
	PRIN201719VMAST01 and from the European Research Council (ERC) under the European Union's Horizon 2020 research and innovation program ERC StG MaMBoQ, n.80290.
	
	\bibliographystyle{plain}
	\bibliography{bibliografia}
	\nocite{*}

\end{document}